\documentclass[a4paper,10pt]{tlp}

\usepackage{tikz}
\usetikzlibrary{trees,decorations.pathmorphing}

\usepackage{float}
\usetikzlibrary{arrows.meta,shapes,positioning,calc,automata,quotes,patterns}
\tikzset{
	>={To[scale length=1.5]},
	intuitionistic relation/.style={->},
	successor relation/.style={->,dashed},
	model state/.style={circle, inner sep=0pt, minimum size=.3cm},
	p model state/.style={model state, color=gray!30!black, fill, text=white},
	not p model state/.style={model state, draw, fill=white},
	big model state/.style={circle, inner sep=0pt, minimum size=.42cm},
	p big model state/.style={big model state, color=gray!30!black, fill, text=white},
	not p big model state/.style={big model state, draw, fill=white}
}

\usepackage[misc,geometry]{ifsym}
\usepackage{amsmath}
\usepackage{amssymb}
\usepackage{graphicx}
\usepackage{caption}
\usepackage{subcaption}
\usepackage{scalerel}
\usepackage{multicol}

\let\next\relax
\RequirePackage{bm}
\RequirePackage{textcomp}
\RequirePackage{upgreek}

\RequirePackage{xcolor}
\makeatletter

\def\starttrans{\xdef\endtrans{\catcode`\noexpand\@=\the\catcode`\@
  \catcode`\noexpand\:=\the\catcode`\:
 }\catcode`\@=11 \catcode`\:=11 }

\starttrans

\newbox\transbox

\newcount\transfactor
\newdimen\trans:dim
\newdimen\trans:dim:a
\newdimen\trans:dim:b
\newdimen\trans:dim:c
\newdimen\trans:dim:d

\newcount\trans:count

\def\trans:def{}
\def\trans:def:a{}
\def\trans:def:b{}
\def\trans:def:c{}
\def\trans:def:d{}

\def\transboxini{\afterassignment\transboxcheck
 \setbox\transbox}

\def\transboxcheck{\ifvoid\transbox
  \expandafter\aftergroup
 \fi\transboxtodo}

\def\transhboxini{\afterassignment\transhboxcheck
 \setbox\transbox}

\def\transhboxcheck{\ifvoid\transbox
  \expandafter\aftergroup\expandafter\transhboxwrap
 \else
  \expandafter\transhboxwrap
 \fi}

\def\transhboxwrap{\ifvbox\transbox
  \setbox\transbox\hbox{\box\transbox}\fi
 \transboxtodo}

\long\def\transboxdef#1\transboxend{\bgroup\def\transboxtodo{#1\egroup}\transboxini}

\long\def\transhboxdef#1\transboxend{\bgroup\def\transboxtodo{#1\egroup}\transhboxini}

\newif\iftransbbox
\newif\iftranscbox
\transbboxtrue
\transcboxtrue

\long\def\transbcboxdef#1\transbboxdef#2\transcboxdef#3\transboxend{\bgroup\edef\transboxtodo{\unexpanded{#1}\iftransbbox\unexpanded{#2}\fi
  \iftranscbox\unexpanded{#3}\else\box\transbox\fi
 \egroup}\transhboxini}

\def\bboxtranson{\hbox\bgroup\transbboxtrue
  \def\transboxtodo{\box\transbox\egroup}\transboxini}

\def\bboxtransoff{\hbox\bgroup\transbboxfalse
  \def\transboxtodo{\box\transbox\egroup}\transboxini}

\def\cboxtranson{\hbox\bgroup\transcboxtrue
  \def\transboxtodo{\box\transbox\egroup}\transboxini}

\def\cboxtransoff{\hbox\bgroup\transcboxfalse
  \def\transboxtodo{\box\transbox\egroup}\transboxini}

\def\boxresizeto#1#2#3{\hbox\transboxdef
  \ifx\relax#1\relax\else \wd\transbox=\dimexpr#1\relax \fi
  \ifx\relax#2\relax\else \ht\transbox=\dimexpr#2\relax \fi
  \ifx\relax#3\relax\else \dp\transbox=\dimexpr#3\relax \fi
  \box\transbox
 \transboxend}

\def\boxresize#1#2#3{\hbox\transboxdef
  \ifx\relax#1\relax\else \wd\transbox=\dimexpr\wd\transbox+#1\relax \fi
  \ifx\relax#2\relax\else \ht\transbox=\dimexpr\ht\transbox+#2\relax \fi
  \ifx\relax#3\relax\else \dp\transbox=\dimexpr\dp\transbox+#3\relax \fi
  \box\transbox
 \transboxend}

\def\boxextents#1#2#3#4{\hbox\transboxdef
  \kern\dimexpr#1\relax
   \vbox{\kern\dimexpr#3\relax
    \box\transbox
    \kern\dimexpr#4\relax
   }\kern\dimexpr#2\relax
 \transboxend}

\def\boxhextent#1#2{\boxextents{#1}{#2}\z@\z@}
\def\boxvextent#1#2{\boxextents\z@\z@{#1}{#2}}

\def\boxexts#1#2#3#4{\hbox\transboxdef
  \trans:dim:a=\wd\transbox
  \trans:dim:b=\dimexpr\trans:dim:a+#1+#2\relax
  \trans:dim:c=\dimexpr\ht\transbox+\dp\transbox\relax
  \trans:dim:d=\dimexpr\trans:dim:c+#3+#4\relax
  \edef\trans:def:a{\fdivide\trans:dim:b\trans:dim:a}\edef\trans:def:b{\fdivide\trans:dim:d\trans:dim:c}\savebp\trans:def:c-\dimexpr#1\relax
  \savebp\trans:def:d\dimexpr-#4+(#3+#4)*\dp\transbox/\trans:dim:c\relax
  \pdfliteral{q \trans:def:a\space 0 0 \trans:def:b\space
                \trans:def:c\space \trans:def:d\space cm}\savebp\trans:def\wd\transbox
  \box\transbox
  \pdfliteral{Q 1 0 0 1 \trans:def\space 0 cm}\transboxend}

\def\boxhext#1#2{\boxexts{#1}{#2}\z@\z@}
\def\boxvext#1#2{\boxexts\z@\z@{#1}{#2}}

\def\boxrevolveleft{\hbox\transbcboxdef
  \trans:dim:a=\wd\transbox
  \trans:dim:b=\ht\transbox
 \transbboxdef
  \wd\transbox=\dimexpr\ht\transbox+\dp\transbox\relax
  \ht\transbox=\trans:dim:a
  \dp\transbox=\z@
 \transcboxdef
  \pdfliteral{q 0 1 -1 0 \tobp{\trans:dim:b} 0 cm}\savebp\trans:def\wd\transbox
  \box\transbox
  \pdfliteral{Q 1 0 0 1 \trans:def\space 0 cm}\transboxend}

\def\boxrevolveright{\hbox\transbcboxdef
  \trans:dim:a=\wd\transbox
  \trans:dim:b=\dp\transbox
 \transbboxdef
  \wd\transbox=\dimexpr\ht\transbox+\dp\transbox\relax
  \ht\transbox=\trans:dim:a
  \dp\transbox=\z@
 \transcboxdef
  \pdfliteral{q 0 -1 1 0 \tobp{\trans:dim:b} \tobp{\trans:dim:a} cm}\savebp\trans:def\wd\transbox
  \box\transbox
  \pdfliteral{Q 1 0 0 1 \trans:def\space 0 cm}\transboxend}

\def\boxrotatexy#1#2#3{\hbox\transboxdef
  \floatsincos\trans:def:a\trans:def:b{#1}\savebp\trans:def:c\dimexpr#2\relax
  \savebp\trans:def:d\dimexpr#3\relax
  \pdfliteral{q \trans:def:b\space
                \negbp\trans:def:a\space
                \trans:def:a\space
                \trans:def:b\space
                \trans:def:c\space
                \trans:def:d\space cm
                1 0 0 1 \negbp\trans:def:c\space
                        \negbp\trans:def:d\space cm}\savebp\trans:def\wd\transbox
  \box\transbox
  \pdfliteral{Q 1 0 0 1 \trans:def\space 0 cm}\transboxend}

\def\boxrotatell#1{\boxrotatexy{#1}\z@{-\dp\transbox}}

\def\boxrotateul#1{\boxrotatexy{#1}\z@{\ht\transbox}}

\newif\ifbboxright

\def\box:rotate:bb#1{\trans:dim:a=\wd\transbox
 \trans:dim:b=\ht\transbox
 \trans:dim:c=\dp\transbox
 \trans:dim:d=\dimexpr\ht\transbox+\dp\transbox\relax
 \trans:count=\reducetrigangle{#1}\fractperiod\relax
 \ifcase\fracttrigfourth\trans:count\relax
  \fr@ct:sin:cos:i\trans:def:a\trans:def:b\trans:count
  \wd\transbox=\dimexpr\fr@ct:mul\trans:dim:a\trans:def:b
                      +\fr@ct:mul\trans:dim:d\trans:def:a\relax
  \ht\transbox=\dimexpr\fr@ct:mul\trans:dim:b\trans:def:b\relax
  \dp\transbox=\dimexpr\fr@ct:mul\trans:dim:a\trans:def:a
                      +\fr@ct:mul\trans:dim:c\trans:def:b\relax
  \savebp\trans:def:c=\dimexpr\fr@ct:mul\trans:dim:c\trans:def:a\relax
 \or
  \fr@ct:sin:cos:ii\trans:def:a\trans:def:b\trans:count
  \wd\transbox=\dimexpr-\fr@ct:mul\trans:dim:a\trans:def:b
                       +\fr@ct:mul\trans:dim:d\trans:def:a\relax
  \ht\transbox=\dimexpr-\fr@ct:mul\trans:dim:c\trans:def:b\relax
  \dp\transbox=\dimexpr-\fr@ct:mul\trans:dim:b\trans:def:b
                        +\fr@ct:mul\trans:dim:a\trans:def:a\relax
  \savebp\trans:def:c=\dimexpr-\fr@ct:mul\trans:dim:a\trans:def:b
                              +\fr@ct:mul\trans:dim:c\trans:def:a\relax
 \or
  \fr@ct:sin:cos:iii\trans:def:a\trans:def:b\trans:count
  \wd\transbox=\dimexpr-\fr@ct:mul\trans:dim:a\trans:def:b
                       -\fr@ct:mul\trans:dim:d\trans:def:a\relax
  \ht\transbox=\dimexpr-\fr@ct:mul\trans:dim:a\trans:def:a
                       -\fr@ct:mul\trans:dim:c\trans:def:b\relax
  \dp\transbox=\dimexpr-\fr@ct:mul\trans:dim:b\trans:def:b\relax
  \savebp\trans:def:c=\dimexpr-\fr@ct:mul\trans:dim:a\trans:def:b
                              -\fr@ct:mul\trans:dim:b\trans:def:a\relax
 \or
  \fr@ct:sin:cos:iv\trans:def:a\trans:def:b\trans:count
  \wd\transbox=\dimexpr\fr@ct:mul\trans:dim:a\trans:def:b
                      -\fr@ct:mul\trans:dim:d\trans:def:a\relax
  \ht\transbox=\dimexpr\fr@ct:mul\trans:dim:b\trans:def:b
                      -\fr@ct:mul\trans:dim:a\trans:def:a\relax
  \dp\transbox=\dimexpr\fr@ct:mul\trans:dim:c\trans:def:b\relax
  \savebp\trans:def:c=\dimexpr-\fr@ct:mul\trans:dim:b\trans:def:a\relax
 \fi
 \ifbboxright
  \trans:dim:d=\dimexpr\fr@ct:mul\trans:dim:a\trans:def:a\relax
  \ht\transbox=\dimexpr\ht\transbox+\trans:dim:d\relax
  \dp\transbox=\dimexpr\dp\transbox-\trans:dim:d\relax
  \savebp\trans:def:d\trans:dim:d
 \else
  \def\trans:def:d{0}\fi
 \edef\trans:def:a{\fr@ct:div\trans:def:a}\edef\trans:def:b{\fr@ct:div\trans:def:b}\pdfliteral{q \trans:def:b\space
               \negbp\trans:def:a\space
               \trans:def:a\space
               \trans:def:b\space
               \trans:def:c\space
               \trans:def:d\space cm}\savebp\trans:def=\wd\transbox
 \box\transbox
 \pdfliteral{Q 1 0 0 1 \trans:def\space 0 cm}}

\bboxrightfalse

\def\box:slant:bb#1#2{\trans:dim:a=\wd\transbox
 \trans:dim:b=\ht\transbox
 \trans:dim:c=\dp\transbox
 \trans:dim:d=\dimexpr\ht\transbox+\dp\transbox\relax
 \trans:count=\reducetrigangle{#1}{2*\fractfourth}\relax
 \ifcase\fracttrigfourth\trans:count\relax
  \fr@ct:sin:cos:i\trans:def:a\trans:def:b\trans:count
  \wd\transbox=\dimexpr\trans:dim:a+\trans:dim:d*\trans:def:a/\trans:def:b\relax
  \savebp\trans:def:c=\dimexpr\trans:dim:c*\trans:def:a/\trans:def:b\relax
 \or
  \fr@ct:sin:cos:ii\trans:def:a\trans:def:b\trans:count
  \wd\transbox=\dimexpr\trans:dim:a-\trans:dim:d*\trans:def:a/\trans:def:b\relax
  \savebp\trans:def:c=\dimexpr-\trans:dim:b*\trans:def:a/\trans:def:b\relax
 \fi
 \edef\trans:def{\fdivide\trans:def:a\trans:def:b}\trans:count=\reducetrigangle{#2}{2*\fractfourth}\relax
 \ifcase\fracttrigfourth\trans:count\relax
  \fr@ct:sin:cos:i\trans:def:a\trans:def:b\trans:count
  \ht\transbox=\dimexpr\trans:dim:b+\trans:dim:a*\trans:def:a/\trans:def:b\relax
 \or
  \fr@ct:sin:cos:ii\trans:def:a\trans:def:b\trans:count
  \dp\transbox=\dimexpr\trans:dim:c-\trans:dim:a*\trans:def:a/\trans:def:b\relax
 \fi
 \ifbboxright
  \trans:dim:d=\dimexpr-\trans:dim:a*\trans:def:a/\trans:def:b\relax
  \ht\transbox=\dimexpr\ht\transbox+\trans:dim:d\relax
  \dp\transbox=\dimexpr\dp\transbox-\trans:dim:d\relax
  \savebp\trans:def:d\trans:dim:d
 \else
  \def\trans:def:d{0}\fi
 \edef\trans:def:a{\fdivide\trans:def:a\trans:def:b}\pdfliteral{q 1
               \trans:def:a\space
               \trans:def\space
               1
               \trans:def:c\space
               \trans:def:d\space cm}\savebp\trans:def=\wd\transbox
 \box\transbox
 \pdfliteral{Q 1 0 0 1 \trans:def\space 0 cm}}

\def\boxxform{\hbox\transboxdef
  \immediate\pdfxform\transbox
  \pdfrefxform\pdflastxform
 \transboxend}

\def\boxxformspec#1\boxxform{\hbox\transboxdef
  \immediate\pdfxform#1\transbox
  \pdfrefxform\pdflastxform
 \transboxend}

\def\boxraise#1{\hbox\transboxdef
  \raise\dimexpr#1\relax\box\transbox
 \transboxend}

\def\boxlower#1{\hbox\transboxdef
  \lower\dimexpr#1\relax\box\transbox
 \transboxend}

\def\boxbaselineat#1{\hbox\transboxdef
  \lower\dimexpr(\ht\transbox+\dp\transbox)*(#1)/\transfactor-\dp\transbox\relax
  \box\transbox
 \transboxend}

\def\boxmoveleft#1{\vbox\transboxdef
  \moveleft\dimexpr#1\relax\box\transbox
 \transboxend}

\def\boxmoveright#1{\vbox\transboxdef
  \moveright\dimexpr#1\relax\box\transbox
 \transboxend}

\def\b@x:rule#1#2#{#1\transboxdef
  \setbox0\hbox{\vrule width\wd\transbox height\ht\transbox depth\dp\transbox #2}\wd\transbox=\wd0
  \ht\transbox=\ht0
  \dp\transbox=\dp0
  \box\transbox
 \transboxend#1}

\def\hboxr{\b@x:rule\hbox}
\def\vboxr{\b@x:rule\vbox}
\def\vtopr{\b@x:rule\vtop}

\def\boxgs#1#2{\hbox\transboxdef
  \pdfliteral{q #1}\savebp\trans:def\wd\transbox
  \box\transbox
  \pdfliteral{#2 Q 1 0 0 1 \trans:def\space 0 cm}\transboxend}

\def\boxmarkers#1#2#3{\hbox\transboxdef
  \copy\transbox
  \trans:dim:a=\dimexpr#1\relax
  \trans:dim:b=\dimexpr#2\relax
  \pdfliteral{q #3}\savebp\trans:def-\dp\transbox
  \box:markers:h
  \savebp\trans:def\ht\transbox
  \box:markers:h
  \savebp\trans:def-\wd\transbox
  \box:markers:v
  \savebp\trans:def\z@
  \box:markers:v
  \pdfliteral{S Q}\setbox\transbox\box\voidb@x
 \transboxend}

\def\box:markers:h{\savebp\trans:def:a\trans:dim:a
  \savebp\trans:def:b\trans:dim:b
  \pdfliteral{\trans:def:a\space\trans:def\space m \trans:def:b\space\trans:def\space l}\savebp\trans:def:a\dimexpr-\wd\transbox-\trans:dim:a\relax
  \savebp\trans:def:b\dimexpr-\wd\transbox-\trans:dim:b\relax
  \pdfliteral{\trans:def:a\space\trans:def\space m \trans:def:b\space\trans:def\space l}}

\def\box:markers:v{\savebp\trans:def:a\dimexpr-\dp\transbox-\trans:dim:a\relax
  \savebp\trans:def:b\dimexpr-\dp\transbox-\trans:dim:b\relax
  \pdfliteral{\trans:def\space \trans:def:a\space m \trans:def\space \trans:def:b\space l}\savebp\trans:def:a\dimexpr\ht\transbox+\trans:dim:a\relax
  \savebp\trans:def:b\dimexpr\ht\transbox+\trans:dim:b\relax
  \pdfliteral{\trans:def\space \trans:def:a\space m \trans:def\space \trans:def:b\space l}}

\def\boxphantom{\hbox\transboxdef
  \hbox to\wd\transbox
   {\vrule width\z@ height\ht\transbox depth\dp\transbox\hss}\transboxend}

\def\boxsmash{\hbox\transhboxdef
  \wd\transbox=\z@
  \ht\transbox=\z@
  \dp\transbox=\z@
  \box\transbox
 \transboxend}

\def\hboxsmash{\hbox\transhboxdef
  \wd\transbox=\z@
  \box\transbox
 \transboxend}

\def\vboxsmash{\vbox\transhboxdef
  \ht\transbox=\z@
  \dp\transbox=\z@
  \box\transbox
 \transboxend}

\def\box:about#1{\hbox\bgroup
  \def\transboxtodo{\trans:dim:a=\wd\transbox
   \trans:dim:b=\ht\transbox
   \trans:dim:c=\dp\transbox
   \box\transbox
   \hbox to\z@{\hss
    \hbox to\trans:dim:a{\hss
     \lower\trans:dim:c\vbox to\z@{\vss
      \vbox to\dimexpr\trans:dim:b+\trans:dim:c{\vss\tt
       #1\vbox{\vskip1ex
        \halign{\hskip1ex plus 1fil####&####\hskip1ex plus 1fil\cr
         \trans:def\span\cr
         wd & \the\trans:dim:a\cr
         ht & \the\trans:dim:b\cr
         dp & \the\trans:dim:c\cr}\vskip1ex
       }\vss}}\hss}}\egroup}\def\trans:def{}\def\trans:def:a{}\box:@bout}

\def\box:@bout#1{\ifcase
  \ifx#1\hbox 0 \else
  \ifx#1\vbox 1 \else
  \ifx#1\vtop 2 \else
  \ifx#1\box  3 \else
  \ifx#1\copy 4 \else 5 \fi\fi\fi\fi\fi
 \edef\trans:def{\trans:def\string\hbox}\expandafter\transboxini\or
 \edef\trans:def{\trans:def\string\vbox}\expandafter\transboxini\or
 \edef\trans:def{\trans:def\string\vtop}\expandafter\transboxini\or
 \edef\trans:def{\trans:def\string\box}\expandafter\boxabout:register\or
 \edef\trans:def{\trans:def\string\copy}\expandafter\boxabout:register\or
 \ifx#1\trans:def:a\errmessage{`#1' is not a box}\fi\let\trans:def:a#1\edef\trans:def{\trans:def\string#1->}\expandafter\expandafter\expandafter\box:@bout\fi
 #1}

\def\boxabout:register#1{\let\trans:def:a#1\afterassignment\boxabout:r@gister\trans:count}

\def\boxabout:r@gister{\edef\trans:def{\trans:def\the\trans:count\space
  (\ifvoid\trans:count void\else
   \ifhbox\trans:count hbox\else
   \ifvbox\trans:count vbox\fi\fi\fi)}\afterassignment\transboxtodo
 \setbox\transbox\trans:def:a\trans:count}

\def\boxabout#1{\box:about{\boxgs{#1}{}}}

\def\expandnumberafter#1#2{\expandafter#1\expandafter{\number#2}}

\def\expandtwonumbersafter#1#2#3{\expandafter#1\expandafter
 {\number#2\expandafter}\expandafter
 {\number#3}}

\def\expandthreenumbersafter#1#2#3#4{\expandafter#1\expandafter
 {\number#2\expandafter}\expandafter
 {\number#3\expandafter}\expandafter
 {\number#4}}

\def\expandnumexprafter#1#2{\expandafter#1\expandafter{\number\numexpr#2}}

\def\expandtwonumexprafter#1#2#3{\expandafter#1\expandafter
 {\number\numexpr#2\expandafter}\expandafter
 {\number\numexpr#3}}

\def\expandthreenumexprafter#1#2#3#4{\expandafter#1\expandafter
 {\number\numexpr#2\expandafter}\expandafter
 {\number\numexpr#3\expandafter}\expandafter
 {\number\numexpr#4}}

\def\expanddimexprafter#1#2{\expandafter#1\expandafter{\the\dimexpr#2}}

\edef\pt:f@ctor{\number\dimexpr100pt} \edef\bp:f@ctor{\number\dimexpr100bp} 

\begingroup
 \catcode`\P=12
 \catcode`\T=12
 \lccode`P=`p
 \lccode`T=`t
 \lowercase{\gdef\with@ut:pt#1PT{#1}}
\endgroup

\def\withoutpt{\expandafter\with@ut:pt}
\def\negbp#1{\withoutpt\the\dimexpr-#1pt\relax}
\def\asbp#1{\withoutpt\the\dimexpr#1*\pt:f@ctor/\bp:f@ctor\relax}

\def\roundbp#1{\expandafter\r@undbp\the\dimexpr(#1)*\pt:f@ctor/\bp:f@ctor\relax0000\relax}
\def\r@undbp{\csname r@und:bp:\the\pdfdecimaldigits\expandafter\endcsname
	\with@ut:pt}

\expandafter\def\csname r@und:bp:0\endcsname #1.#2#3\relax{\number\numexpr#1#2/10\relax}
\expandafter\def\csname r@und:bp:1\endcsname #1.#2#3#4\relax{\round:bp:once{#1}{#2#3}\relax}
\expandafter\def\csname r@und:bp:2\endcsname #1.#2#3#4#5\relax{\round:bp:once{#1}{#2#3#4}\relax}
\expandafter\def\csname r@und:bp:3\endcsname #1.#2#3#4#5#6\relax{\round:bp:once{#1}{#2#3#4#5}\relax}
\expandafter\def\csname r@und:bp:4\endcsname #1.#2#3#4#5#6#7\relax{\round:bp:once{#1}{#2#3#4#5#6}\relax}

\def\round:bp:once#1#2{\ifnum#11<0-\number\numexpr-\else\number\numexpr\fi
 #1+(\m@ne+\expandafter\r@und:bp:once\number\numexpr1#2/10\relax}

\def\r@und:bp:once#1#2\relax{#1)\relax\ifnum#2>0.#2\fi}

\def\set:bp:rounder#1#2{\expandafter\edef\csname #1:\the\numexpr#2\relax\endcsname##1{\unexpanded{\expandafter\expandafter\expandafter}\expandafter\noexpand
  \csname r@und:bp:\the\numexpr#2\relax\endcsname
  \unexpanded{\expandafter\with@ut:pt\the}\dimexpr(##1)*\unexpanded{\pt:f@ctor/\bp:f@ctor}\relax0000\relax}}

\set:bp:rounder{roundbpto}{0}
\set:bp:rounder{roundbpto}{1}
\set:bp:rounder{roundbpto}{2}
\set:bp:rounder{roundbpto}{3}
\set:bp:rounder{roundbpto}{4}

\def\roundbpto#1{\csname roundbpto:#1\endcsname}

\def\enablebpround{\let\tobp\roundbp}
\def\disablebpround{\let\tobp\asbp}
\def\setbpround#1{\expandafter\let\expandafter\tobp\csname roundbpto:\the\numexpr#1\relax\endcsname}

\enablebpround

\def\savebp#1{\def\s@vebp{\edef#1{\tobp{\bp:dim@n}}}\afterassignment\s@vebp\bp:dim@n}

\newdimen\bp:dim@n

\def\absoluteint#1{\numexpr\ifnum#1<\z@-\fi#1}

\def\absolutedim#1{\dimexpr\ifdim#1<\z@-\fi#1}

\def\expanddivisionafter#1#2#3{\expandnumexprafter#1{#2/#3}{#2}{#3}}

\def\divfloor{\expandtwonumexprafter\dividefloor}
\def\dividefloor{\expanddivisionafter\divide:fl@@r}
\def\divide:fl@@r#1#2#3{\numexpr#1\ifcase\ifnum#2<0 \ifnum#3<0 1 \else 0 \fi
        \else      \ifnum#3<0 1 \else 0 \fi \fi
 \ifnum\numexpr#1*#3>#2-\@ne\fi\or
 \ifnum\numexpr#1*#3<#2-\@ne\fi\fi}

\def\divceil{\expandtwonumexprafter\divideceil}
\def\divideceil{\expanddivisionafter\divide:c@il}
\def\divide:c@il#1#2#3{\numexpr#1\ifcase\ifnum#2<0 \ifnum#3<0 1 \else 0 \fi
        \else      \ifnum#3<0 1 \else 0 \fi \fi
 \ifnum\numexpr#1*#3<#2+\@ne\fi\or
 \ifnum\numexpr#1*#3>#2+\@ne\fi\fi}

\def\divint{\expandtwonumexprafter\divideint}
\def\divideint{\expanddivisionafter\divide:int}
\def\divide:int#1#2#3{\numexpr#1\ifcase\ifnum#2<0 \ifnum#3<0 3 \else 1 \fi
        \else      \ifnum#3<0 2 \else 0 \fi \fi
 \ifnum\numexpr#1*#3>#2-\@ne\fi\or
 \ifnum\numexpr#1*#3<#2+\@ne\fi\or
 \ifnum\numexpr#1*#3>#2+\@ne\fi\or
 \ifnum\numexpr#1*#3<#2-\@ne\fi\fi}

\def\divnint{\expandtwonumexprafter\dividenint}
\def\dividenint#1#2{\numexpr#1/#2}

\def\modulo{\expanddivisionafter\do:m@dulo}
\def\do:m@dulo#1#2#3{\numexpr#2-#3*\divide:fl@@r{#1}{#2}{#3}\relax}

\def\dividefloorpos{\expanddivisionafter\divide:fl@@r:pos}
\def\divide:fl@@r:pos#1#2#3{\numexpr#1\ifnum\numexpr#1*#3>#2-\@ne\fi}

\def\divide:c@il:pos#1#2#3{\numexpr#1\ifnum\numexpr#1*#3<#2+\@ne\fi}

\def\modpos{\expandtwonumexprafter\modulopos}
\def\modulopos{\expanddivisionafter\modulo:p@s}
\def\modulo:p@s#1#2#3{\numexpr#2-#3*\divide:fl@@r:pos{#1}{#2}{#3}\relax}

\def\floatround#1{\divnint{\dimexpr#1pt}\p@}
\def\floatfloor#1{\divfloor{\dimexpr#1pt}\p@}
\def\floatceil#1{\divceil{\dimexpr#1pt}\p@}
\def\floatint#1{\divint{\dimexpr#1pt}\p@}
\def\floatnint#1{\divnint{\dimexpr#1pt}\p@}

\def\fdivide{\expandtwonumexprafter\flo@t:divide}
\def\flo@t:divide#1#2{\withoutpt\the\dimexpr\numexpr#1*\p@/#2\relax sp\relax}

\def\divfloat{\expandthreenumexprafter\dividefloat}

\def\dividefloat#1#2#3{\expandnumberafter\divide:flo@t {\absoluteint{\divideint{#1}{#2}}}{#1}{#2}{#3}}

\def\divide:flo@t#1#2#3{\ifnum#2<0 \ifnum#3>0 -\fi\else
 \ifnum#2>0 \ifnum#3<0 -\fi\fi\fi
 #1.\expandthreenumexprafter\divide:fl@@t
 {#1}{\absoluteint{#2}}{\absoluteint{#3}}}

\def\divide:fl@@t#1#2#3{\expandnumexprafter\divide:flo@t:modulo{#2-#1*#3}{#3}}

\def\divide:flo@t:modulo#1#2{\ifnum#1<214748365
  \expandtwonumbersafter\divide:flo@t:result{#10}{#2\expandafter}\else
  \expandtwonumexprafter\divide:flo@t:modulo{#1/2}{#2/2\expandafter}\fi}

\def\divide:flo@t:result#1#2#3{\ifnum#3>1
  \expandtwonumexprafter\divide:flo@t:repeat
  {#3-\@ne}{\dividefloorpos{#1}{#2}\expandafter}\else
  \number\divide:flo@t:last{#1}{#2}\relax
  \expandafter\gobbletwo
 \fi{#1}{#2}}

\def\divide:flo@t:repeat#1#2#3#4{#2\expandnumexprafter\divide:flo@t:modulo{#3-#2*#4}{#4}{#1}}

\newcount\floatprecision
\def\roundlast{\let\divide:flo@t:last\dividenint}
\def\floorlast{\let\divide:flo@t:last\divideint}
\roundlast

\def\tofixedbp#1{\divfloat{\dimexpr#1}\b@\floatprecision}
\def\roundfixedbp#1{\divfloat{\dimexpr#1}\b@\pdfdecimaldigits}

\def\fractdegree#1{\numexpr16*\dimexpr#1pt}           \edef\fractfactor{\number\numexpr\maxdimen+\@ne}      \edef\fractfourth{\number\numexpr90*\fractdegree\@ne} \edef\fractperiod{\number\numexpr4*\fractfourth}      

\def\reducefractangle#1{\expandnumberafter\reduce:fr@ct:angle{\fractdegree{#1}}}

\def\reduce:fr@ct:angle#1#2{\ifnum#1<0
  \numexpr#2-\modpos{-#1}{#2}\relax
 \else
  \modpos{#1}{#2}\fi}

\def\reduceintangle#1#2{\expandtwonumexprafter\reduce:int:@ngle{#1}{#2/\fractdegree\@ne}}

\def\reduce:int:@ngle#1#2{\fractdegree{\reduce:fr@ct:angle{#1}{#2}}}

\def\enablefractangle{\let\reducetrigangle\reducefractangle}
\def\disablefractangle{\let\reducetrigangle\reduceintangle}
\enablefractangle

\def\fracttrigfourth#1{\dividefloorpos{#1}\fractfourth}

\def\fr@ct:mul#1#2{#1*#2/\fractfactor}
\def\fr@ct:div#1{\fdivide{#1}\fractfactor}

\def\fr@ct:angle#1{\ifcase\numexpr#1\relax
62914560\or 47185920\or 31457280\or 16777216\or 8388608\or 4194304\or 2097152\or 1048576\or 524288\or 262144\or 131072\or 65536\or 32768\or 16384\or 8192\or 4096\or 2048\or 1024\or 512\or 256\or 128\or 64\or 32\or 16\or 8\or 4\or 2\or 1\fi}

\def\fr@ct:sin#1{\ifcase\numexpr#1\relax
929887697\or 759250125\or 536870912\or 295963357\or 149435979\or 74900443\or 37473049\or 18739379\or 9370046\or 4685068\or 2342539\or 1171270\or 585635\or 292818\or 146409\or 73204\or 36602\or 18301\or 9151\or 4575\or 2288\or 1144\or 572\or 286\or 143\or 71\or 36\or 18\fi}

\def\fr@ct:cos#1{\ifcase\numexpr#1\relax
 536870912\or 759250125\or 929887697\or 1032146887\or 1063292242\or 1071126243\or 1073087729\or 1073578288\or 1073700939\or 1073731603\or 1073739269\or 1073741185\or 1073741664\or 1073741784\or 1073741814\or 1073741822\or 1073741823\or 1073741824\or 1073741824\or 1073741824\or 1073741824\or 1073741824\or 1073741824\or 1073741824\or 1073741824\or 1073741824\or 1073741824\or 1073741824\fi}

\trans:count=0
\loop
 \expandafter\edef\csname
   fractangle:\the\trans:count\endcsname{\fr@ct:angle\trans:count}
 \expandafter\edef\csname
   fractsinvalue:\the\trans:count\endcsname{\fr@ct:sin\trans:count}
 \expandafter\edef\csname
   fractcosvalue:\the\trans:count\endcsname{\fr@ct:cos\trans:count}
 \ifnum\trans:count<27
 \advance\trans:count by1
\repeat
\def\fr@ct:angle#1{\csname fractangle:\number#1\endcsname}
\def\fr@ct:sin#1{\csname fractsinvalue:\number#1\endcsname}
\def\fr@ct:cos#1{\csname fractcosvalue:\number#1\endcsname}

\def\fracttrig#1{\expandnumexprafter\fr@ct:trig{\reducetrigangle{#1}\fractperiod}}

\def\fr@ct:trig#1{\csname fr@ct:trig:\romannumeral\fracttrigfourth{#1}+\@ne\endcsname
 {#1}}

\def\fr@ct:trig:i#1#2#3{\expandthreenumexprafter\fr@ct:trig:cont{#1}{#2}{#3}\z@}

\def\fr@ct:trig:ii#1#2#3{\expandthreenumexprafter\fr@ct:trig:cont{#1-\fractfourth}{#3}{-#2}\z@}

\def\fr@ct:trig:iii#1#2#3{\expandthreenumexprafter\fr@ct:trig:cont{#1-2*\fractfourth}{-#2}{-#3}\z@}

\def\fr@ct:trig:iv#1#2#3{\expandthreenumexprafter\fr@ct:trig:cont{#1-3*\fractfourth}{-#3}{#2}\z@}

\def\fr@ct:trig:cont#1#2#3#4{\ifcase
  \ifnum#1>0 \ifnum#1<\fr@ct:angle{#4} 0 \else 1 \fi \else 2 \fi
  \expandafter\fr@ct:trig:cont\expandafter
  {\number#1\expandafter}\expandafter
  {\number#2\expandafter}\expandafter
  {\number#3\expandafter}\expandafter
  {\number\numexpr#4+\@ne\expandafter}\or
  \expandafter\fr@ct:trig:cont\expandafter
  {\number\numexpr#1-\fr@ct:angle{#4}\expandafter}\expandafter
  {\number\numexpr\fr@ct:mul{#2}{\fr@ct:cos{#4}}+\fr@ct:mul{#3}{\fr@ct:sin{#4}}\expandafter}\expandafter
  {\number\numexpr\fr@ct:mul{#3}{\fr@ct:cos{#4}}-\fr@ct:mul{#2}{\fr@ct:sin{#4}}\expandafter}\expandafter
  {\number\numexpr#4+\@ne\expandafter}\or
  \fracttrigend{#2}{#3}\fi}

\def\fracttrigend#1#2\fi#3{\fi#3{#1}{#2}}

\def\fractsincos#1#2#3{\fracttrig{#3}\z@\fractfactor\fr@ct:sin:cos#1#2}
\def\fr@ct:sin:cos#1#2#3#4{\def#3{#1}\def#4{#2}}

\def\fr@ct:sin:cos:i#1#2#3{\fr@ct:trig:i{#3}\z@\fractfactor\fr@ct:sin:cos#1#2}
\def\fr@ct:sin:cos:ii#1#2#3{\fr@ct:trig:ii{#3}\z@\fractfactor\fr@ct:sin:cos#1#2}
\def\fr@ct:sin:cos:iii#1#2#3{\fr@ct:trig:iii{#3}\z@\fractfactor\fr@ct:sin:cos#1#2}
\def\fr@ct:sin:cos:iv#1#2#3{\fr@ct:trig:iv{#3}\z@\fractfactor\fr@ct:sin:cos#1#2}

\def\floatsincos#1#2#3{\fracttrig{#3}\z@\fractfactor\flo@t:sin:c@s#1#2}
\def\flo@t:sin:c@s#1#2#3#4{\edef#3{\fr@ct:div{#1}}\edef#4{\fr@ct:div{#2}}}

\endtrans
 \newbox\qbox
\def\usecolor#1{\csname\string\color@#1\endcsname\space}
\newcommand\bordercolor[1]{\colsplit{1}{#1}}
\newcommand\fillcolor[1]{\colsplit{0}{#1}}
\newcommand\outline[1]{\leavevmode \def\maltext{#1}\setbox\qbox=\hbox{\maltext}\boxgs{Q q 2 Tr \thickness\space w \fillcol\space \bordercol\space}{}\copy\qbox }
\makeatother
\newcommand\colsplit[2]{\colorlet{tmpcolor}{#2}\edef\tmp{\usecolor{tmpcolor}}\def\tmpB{}\expandafter\colsplithelp\tmp\relax \ifnum0=#1\relax\edef\fillcol{\tmpB}\else\edef\bordercol{\tmpC}\fi}
\def\colsplithelp#1#2 #3\relax{\edef\tmpB{\tmpB#1#2 }\ifnum `#1>`9\relax\def\tmpC{#3}\else\colsplithelp#3\relax\fi
}
\bordercolor{black}
\fillcolor{white}
\def\thickness{.3}

\newcommand{\next}{\text{\rm \raisebox{-.5pt}{\Large\textopenbullet}}}    

\newcommand{\alwaysF}{\ensuremath{\square}}

\newcommand{\eventuallyF}{\ensuremath{\Diamond}}

\newcommand{\until}{\ensuremath{\mathbin{\mbox{\outline{$\bm{\mathsf{U}}$}}}}}
\newcommand{\release}{\ensuremath{\mathbin{\mbox{\outline{$\bm{\mathsf{R}}$}}}}}

\mathchardef\mhyphen="2D

\newcommand{\intervcc}[2]{\ensuremath{[#1..#2]}}
\newcommand{\intervco}[2]{\ensuremath{[#1..#2)}}

\newcommand{\rangecc}[3]{\ensuremath{#1 \in \intervcc{#2}{#3}}}
\newcommand{\rangeco}[3]{\ensuremath{#1 \in \intervco{#2}{#3}}}

 \providecommand{\logfont}{\textrm}

\newcommand{\HT}{\ensuremath{\logfont{HT}}}
\newcommand{\EL}{\ensuremath{\logfont{EL}}}

\newcommand{\LTL}{\ensuremath{\logfont{LTL}}}

\newcommand{\THT}{\ensuremath{\logfont{THT}}}

\newcommand{\TEL}{\ensuremath{\logfont{TEL}}}

\newcommand{\tuple}[1]{\ensuremath{\langle #1 \rangle}}

\newcommand{\M}{\ensuremath{\mathbf{M}}}

\newcommand{\model}{{\mathfrak M}}

\newcommand{\parts}[1]{\ensuremath{2^{#1}}} \newcommand{\deffun}[3]{#1 : #2 \rightarrow #3} \newcommand{\Th}[1]{\ensuremath{\mathit{Th}(#1)}}
\newcommand{\msuccfct}[1][]{S#1}

\newcommand{\irel}{\preccurlyeq}            
\newcommand{\sfun}{\msuccfct}
\newcommand{\peq}{\preccurlyeq}
\newcommand{\langP}{\ensuremath{\mathcal{L}_{p}}} \newcommand{\langT}{\ensuremath{\mathcal{L}_{t}}} 

\newenvironment{proofof}[1]{\noindent {\bf Proof of #1.}}{\bigskip}
\newcommand{\PV}{\ensuremath{\mathbb{P}}}
\newcommand{\myatom}{\ensuremath{p}}

\newcommand{\eqdef}{\ensuremath{:=}}

\def\I{\logfont{INT}}
\def\C{\logfont{CL}}
\def\KC{\logfont{KC}}
\newcommand{\logic}[1]{\ensuremath{\logfont{#1}}}
\newcommand{\iltl}{\ensuremath{\mathrm{ITL^e}}}
\newcommand{\itlb}{\ensuremath{\mathrm{ITL^p}}}

\newcommand{\itlbd}[1]{\ensuremath{\mathrm{ITL^{\BD{#1}}}}}

\newcommand{\Inf}[1]{\ensuremath{\models_{\scaleto{#1}{3pt}}}} \newcommand{\TPV}{\ensuremath{\PV^{\next}}}

\newcommand{\qed}{\ignorespaces} \newtheorem{definition}{Definition}
\newtheorem{theorem}{Theorem}
\newtheorem{proposition}{Proposition}
\newtheorem{corollary}{Corollary}
\newtheorem{lemma}{Lemma}
\newtheorem{observation}{Observation}

\newcommand{\cn}[2]{\ensuremath{Cn_{#1}(#2)}}

\newcommand{\bd}[1]{\ensuremath{\bm{bd_{#1}}}}
\newcommand{\BD}[1]{\ensuremath{\mathrm{BD_{#1}}}}
\newcommand{\depth}[1]{\ensuremath{\mathit{depth}(#1)}} \newcommand{\FName}[1]{\ensuremath{\mathfrak{#1}}}

\newcommand{\bisim}[0]{\ensuremath{\mathcal{Z}}}

\begin{document}

\lefttitle{Cambridge Author}
\jnlPage{1}{8}
\jnlDoiYr{2021}
\doival{10.1017/xxxxx}

\title[Intuitionistic Temporal Logic and Temporal ASP]{Applications of Intuitionistic Temporal Logic to Temporal Answer Set Programming\footnote{This paper is an extended version of~\citep{CabalarDLSS24}, presented at the 17th International Conference on Logic Programming and Nonmonotonic Reasoning (LPNMR 2024).}}

\begin{authgrp}
	\author{\sn{Pedro} \gn{Cabalar}}
	\affiliation{University of Corunna, Spain}
	\author{\sn{Mart\'{\i}n} \gn{Di\'eguez}}
	\affiliation{University of Angers, France}
	\author{\sn{David} \gn{Fern\'andez-Duque}}
	\affiliation{University of Barcelona, Spain}
	\author{\sn{Fran\c{c}ois} \gn{Laferri\`ere}}
        \affiliation{Potassco Solutions, Germany}
	\author{\sn{Torsten} \gn{Schaub}}
	\affiliation{University of Potsdam, Germany}\affiliation{Potassco Solutions, Germany}
	\author{\sn{Igor} \gn{St\'ephan}}
	\affiliation{University of Angers, France}
\end{authgrp}

\maketitle

\begin{abstract}
  The relationship between intuitionistic or intermediate logics and logic programming has been extensively studied,
  prominently featuring Pearce's equilibrium logic and Osorio's safe beliefs.
  Equilibrium logic admits a fixpoint characterization based on the logic of here-and-there,
  akin to theory completion in default and autoepistemic logics.
  Safe beliefs are similarly defined via a fixpoint operator,
  albeit under the semantics of intuitionistic or other intermediate logics.

  In this paper,
  we investigate the logical foundations of Temporal Answer Set Programming through the lens of Temporal Equilibrium Logic,
  a formalism combining equilibrium logic with linear-time temporal operators.
  We lift the seminal approaches of Pearce and Osorio to the temporal setting,
  establishing a formal correspondence between temporal intuitionistic logic and temporal logic programming.
  Our results deepen the theoretical underpinnings of Temporal Answer Set Programming and
  provide new avenues for research in temporal reasoning.

\vspace{5pt}
Under consideration in Theory and Practice of Logic Programming (TPLP).
\end{abstract}
 
\section{Introduction}
Temporal logic programming, introduced in the late 1980s~\citep{abadimanna89},
augments logic programming with temporal modal operators,
primarily those from linear-time temporal logic~(\LTL;~\citealp{pnueli77a}).
Although this field experienced substantial research activity throughout the 1980s and 1990s, its momentum eventually waned.
More recently, the advent of answer set programming (ASP; \citealp{lifschitz19a}),
and particularly its demonstrated efficacy in modeling and resolving temporal scenarios,
has sparked renewed interest in these foundational approaches to temporal logic programming.

Early approaches to time representation in ASP~\citep{gellif93a} relied on variables
ranging over finite subsets of the natural numbers.
Although straightforward,
this methodology lacked the dedicated language constructs and specialized inference mechanisms characteristic of \LTL.
Consequently, it remains infeasible to represent or reason about properties of reactive systems over infinite traces,
such as
safety (e.g., ``Is a specific state reachable?'') or
liveness (e.g., ``Does a condition hold infinitely often?'').
Furthermore, establishing the unsolvability of planning problems becomes substantially more difficult.

To overcome these limitations, several extensions of ASP with temporal operators have been investigated.
For instance,
\cite{eitsim09a} extended logic programs with function symbols to model both past and future temporal references.
Other modal-inspired approaches typically adopt a temporal or dynamic modal logic~\citep{pnueli77a,hatiko00a} as a monotonic basis,
subsequently introducing nonmonotonicity via an established ASP semantics~\citep{lifschitz10a}.
As an example,
\cite{gimath13a} generalized the traditional reduct-based semantics~\citep{gellif88b} to a logic programming fragment
equipped with dynamic logic operators~\citep{hatiko00a}.
Similarly,
\cite{agcadipevi13a} integrated \LTL\ with equilibrium logic (\EL;~\citealp{pearce06a}),
the predominant logical characterization of stable models and answer sets.
This latter framework was subsequently adapted to finite traces in~\citep{agcadipescscvi20a}.

Equilibrium logic builds upon here-and-there logic (\HT;~\citealp{heyting30a})
by imposing a minimal model selection criterion to capture answer sets.
In his seminal work, \cite{pearce06a} provided an alternative formulation of equilibrium logic based on fixpoints.
Analogous to the treatment of default and autoepistemic logics~\citep{martru93a},
this characterization relies on theory extensions, or completions, rather than direct semantic minimization.

A related fixpoint characterization, termed safe beliefs, was introduced by~\cite{osnaar05a}.
Instead of \HT, their approach employs intuitionistic logic (\I;~\citealp{mints20a}) as its monotonic foundation.
Crucially, they demonstrated that \I\ can be substituted with any (intermediate) logic $X$
satisfying $\I \subseteq \logic{X} \subseteq \HT$
without altering the resulting safe beliefs.
Consequently, safe beliefs provide a robust framework
for investigating properties of (temporal) logic programs from a broader logical standpoint,
facilitating novel and insightful program transformations.

In this paper, we extend both Pearce's and Osorio's fixpoint-based characterizations to the temporal domain.
Regarding Pearce's approach, we demonstrate that theory completions coincide with temporal equilibrium models
when \HT\ is superseded by Temporal Here-and-There logic (\THT; \citep{baldie16}).

Extending Osorio's approach to the temporal setting presents several notable challenges.

First, Osorio's work heavily relies on fundamental properties of propositional intuitionistic and intermediate logics~\citep{Gabbay81}.
One pivotal property is that satisfiability (or consistency) in intuitionistic logic is preserved across all intermediate logics.
Furthermore, if a formula is satisfiable, it is guaranteed to hold in a finite model.
Unfortunately, these properties generally fail in the temporal case.
To circumvent this, we identify an intuitionistic temporal logic that preserves these characteristics,
serving as the ``weakest'' intuitionistic base logic in our framework.

Second, Osorio's approach depends on a syntactic consequence relation and Hilbert-style axiomatic systems for intuitionistic logic.
However, to the best of our knowledge, no sound and complete axiomatic system currently exists for an intuitionistic version of
\LTL.\footnote{In~\citep{bodifero19a},
  four axiomatic systems for intuitionistic temporal logics are studied.
  It is shown that each generates a different logic, inducing distinct confluence properties.
  However, their completeness remains unaddressed.
  A sound and complete axiomatic system for the $\alwaysF$-free fragment of \iltl\ is presented in~\citep{DieguezF18}.}
Alternative attempts introduce the co-implication connective~\citep{Rauszer1974-RAUAFO} into the
language~\citep{femcze24,ADFM25,0001DFM22}.
Additionally, Osorio's method employs syntactic transformations to eliminate propositional variables.
These transformations cannot be directly lifted to the temporal setting,
as the truth values of propositional variables vary dynamically over time.

In light of these challenges, we adopt a strictly semantic approach.
We reformulate Osorio's fundamental results using a semantic entailment relation,
which we subsequently generalize to the temporal case.
A critical component of this strategy is the application of bisimulations for both
intuitionistic~\citep{patterson97a} and intuitionistic temporal logics~\citep{BalbianiBDF20}.

Beyond reformulating Osorio's results semantically, we formally define the notion of an $\logic{X}$-temporal safe
belief set, where $\logic{X}$ is any intermediate temporal logic.
We establish two primary results:
first, that \THT-temporal safe belief sets exactly correspond to temporal equilibrium models; and
second, that substituting \THT\ with any weaker intermediate temporal logic yields the identical set of safe beliefs.

The remainder of this paper is organized as follows.
Section~\ref{sec:il} provides the background on propositional intuitionistic and equilibrium logics.
Section~\ref{sec:fixpoint} introduces Pearce's theory completions alongside our semantic reformulation of Osorio's safe beliefs.
Section~\ref{sec:itl} reviews intuitionistic and intermediate temporal logics,
detailing the specific technical results employed in our framework.
Section~\ref{sec:main-contribution} presents our primary contribution:
lifting Pearce's theory completions and Osorio's safe beliefs to the temporal domain.
Finally, we conclude the paper with a brief discussion and directions for future research.

 \section{Intuitionistic and Intermediate Logics}\label{sec:il}
Given a countable, possibly infinite set \PV\ of atoms, also called \emph{alphabet},
our basic language $\langP$ consists of formulas generated by the following grammar:
\begin{displaymath}
\varphi ::= \myatom\in\PV \mid \bot \mid \varphi\wedge\varphi \mid \varphi\vee\varphi \mid \varphi\to\varphi
\end{displaymath}
The negation connective is defined in terms of implication as $\neg \varphi \eqdef \varphi \to \bot$.
A \emph{(propositional) theory} is a possibly infinite set of propositional formulas.

Formulas of $\langP$ are interpreted over partially ordered sets.
An \emph{intuitionistic frame} is a tuple $\FName{F}=(W,\peq)$,
where $W$ is a non-empty set of (Kripke) worlds and
${\irel}\subseteq {W\times W}$ is a partial order.
Given a frame $\FName{F}=(W,\peq)$,
a world $w\in W$ is \emph{$\peq$-maximal} if there is no $v\in W$ such that $w\not=v$ and $w \peq v$.
Because maximality is exclusively associated with the relation $\peq$ throughout this work, we simply use the term \emph{maximal}.

Given a frame $\FName{F} = (W,\peq)$, we say that a subset $U\subseteq W$ is an \emph{upset} of $\FName{F}$
if for every $w,v\in W$ we have that if both $w\in U$ and $w \peq v$ then $v \in U$.
Moreover, a frame $\FName{F'}=(U,\peq')$ is called a \emph{generated subframe} of $\FName{F}$
if $U \subseteq W$ is an upset of $\FName{F}$ and $\peq'$ is the restriction of $\peq$ to $U$, that is,
${\peq'} = {\peq \cap {(U \times U)}}$.
Finally, given $w \in W$, we define the subframe generated by $w$ as the
generated subframe $\FName{F'} = (\lbrace u \in W\mid w \peq u \rbrace, \peq')$.

We say that an intuitionistic frame $\FName{F} = (W,\peq)$ is of \emph{depth} $n$, $\depth{\FName{F}}=n$ in symbols,
if there is a chain of $n$ worlds in $\FName{F}$ and no chain of more than $n$ worlds.
Whenever $\FName{F}$ contains an $n$-world chain for every $n<\omega$,
we say that $\FName{F}$ is of \emph{infinite depth} $\infty$.
Given a frame $\FName{F} =(W,\peq)$ and $w \in W$,
we denote the \emph{depth} of the subframe generated by $w$ as $\depth{(\FName{F},w)}$.

An \emph{intuitionistic model}, or simply \emph{model},
is a tuple $\model= \tuple{(W,\peq),V}$ consisting of a frame $(W,\peq)$ equipped with a monotone valuation function
\(
\deffun V   W {\parts\PV}
\).
That is, if $w \irel v$, then $V(w) \subseteq V(v)$ for all $w, v \in W$.
The satisfaction relation (denoted by $\models$) of a formula $\varphi$ at $w \in W$ is defined inductively by:
\begin{enumerate}
\item $\model, w \models p $  iff $p \in V(w) $
\item $\model, w \not \models \bot$
\item $\model, w \models \varphi\wedge \psi$  iff $  \model, w \models \varphi $ and $\model, w \models \psi$
\item $\model, w \models \varphi \vee \psi $  iff  $ \model, w \models \varphi $ or $\model, w \models \psi$
\item $\model, w \models \varphi \rightarrow \psi $ iff for all $v \succcurlyeq w $, if $\model, v \models \varphi$, then $\model, v \models \psi $
\end{enumerate}
A formula $\varphi$ is satisfied in an intuitionistic model $\model = \tuple{\FName{F},V}$,
in symbols $\model \models \varphi$,
if $\model, w \models \varphi$ for some $w \in \FName{F}$.
A formula $\varphi$ is satisfied on an intuitionistic frame $\FName{F}$,
if there exists a model $\model=\tuple{\FName{F},V}$ such that $\model \models \varphi$.
A formula $\varphi$ is valid on an intuitionistic frame $\FName{F}$,
in symbols $\FName{F} \models \varphi$,
if for all models $\model=\tuple{\FName{F},V}$, we have $\model\models \varphi$.
In the case of a theory $\Gamma$, we say that $\model, w \models \Gamma$ if $\model, w \models \varphi$
for all $\varphi \in \Gamma$.
Similarly, $\Gamma$ is said to be \emph{consistent},
if there is a model $\model$ and a world $w$ such that $\model, w \models \Gamma$.
Finally we define the \emph{intuitionistic logic} as
\begin{equation*}
  \I := \lbrace  \varphi \in \langP \mid \FName{F}\models \varphi\rbrace,
\end{equation*}
where $\FName{F}$ is an intuitionistic frame.

\subsection{Intermediate Logics}
An \emph{intermediate logic}\footnote{\cite{chagrov} made a distinction between \emph{super-intuitionistic} and intermediate logics but they also
  mention that in the propositional case, these two notions are practically identical.}
in the language \langP\ is any set of formulas $\logic{X}$ satisfying the following conditions:
\begin{enumerate}[label=\arabic*)]
\item $\I \subseteq \logic{X}\subseteq \C$, where $\C$ stands for \emph{classical logic},
\item $\logic{X}$ is closed under \emph{modus ponens},
  i.e., $\varphi, \varphi \to \psi \in \logic{X}$ implies $\psi \in \logic{X}$,
\item $\logic{X}$ is closed under \emph{uniform substitution},
  i.e., $\varphi \in \logic{X}$ implies $\varphi\mathbf{s} \in \logic{X}$ for any $\varphi \in \langP$ and substitution
  $\mathbf{s}$.\footnote{A \emph{substitution} $\mathbf{s}$ is a mapping $\mathbf{s}:\PV \to \langP$ and
    $\varphi\mathbf{s}$ is defined by induction on the construction of $\varphi$: $\varphi \mathbf{s} = \mathbf{s}(p)$,
    $\bot\mathbf{s} = \bot$ and $(\varphi \odot \psi)\mathbf{s} = \varphi\mathbf{s} \odot \psi \mathbf{s}$
    for $\odot \in \lbrace \wedge, \vee, \rightarrow \rbrace$.}
\end{enumerate}

A \emph{proper intermediate logic} is an intermediate logic different from $\C$. Broadly speaking, intermediate logics
are obtained by adding formulas (that are classically valid) to $\I$ as axiom schemas~\cite[Chapter 2]{Gabbay81}.
In this way, they impose restrictions on $\peq$.
Therefore, given an intermediate logic $\logic{X}$,
we define the \emph{class} of $\logic{X}$-frames as the set of all intuitionistic frames
$(W,\peq)$ where $\peq$ satisfies the restriction induced by the schemas used to generate $\logic{X}$.
As in the intuitionistic case, an intermediate logic $\logic{X}$ can be defined as
\begin{equation*}
	\logic{X} := \lbrace  \varphi \in \langP \mid \FName{F}\models \varphi\rbrace,
\end{equation*}
where $\FName{F}$ is an $\logic{X}$-frame.

To give an example,
the logic of the \emph{weak exclude middle}~\citep{Jankov68,Gabbay81} ($\KC$) is obtained by
adding the axiom $\neg p \vee \neg \neg p$ to $\I$ and it is
characterized by intuitionistic frames $(W,\peq)$ satisfying the following frame condition:
there exists $u\in W$ such that $v \peq u$ for all $v\in W$.\footnote{In the literature, this frame condition is usually called \emph{topwidth 1}~\citep{Gabbay81}.}
Another family of intermediate logics, denoted by $\BD{n}$, are obtained by adding an instance of the axiom schema $\bd{n}$ to $\I$.
For a given $n \ge 1$, such a family of axioms is recursively defined as follows:
\begin{eqnarray*}
  \bd{1} &:=& p_1\vee \neg p_1\\
  \bd{n+1} & := & p_{n+1}\vee \left(p_{n+1}\to \bd{n}\right).
\end{eqnarray*}
The axiom $\bd{n}$ induces the following property on intuitionistic frames.
\begin{theorem}[\citealp{chagrov}]\label{depth:axiom}
  An intuitionistic frame $\FName{F}=(W,\peq)$ validates $\bd{n}$ iff $\depth{\FName{F}}\le n$,
  i.e, iff $\FName{F}$ satisfies the following condition
  \begin{equation}
    \forall w_0,\cdots, \forall w_n \left( \left(\bigwedge\limits_{i=0}^{n-1} w_i \peq w_{i+1}\right) \to \bigvee\limits_{i\not=j}\left(w_i = w_j\right)\right).\label{eq:int:depth}
  \end{equation}
\end{theorem}

The strongest proper intermediate logic is the logic of here-and-there (\HT),
which has been studied in the literature by different authors~\citep{heyting30a,goedel32a,smetanich60}.
This logic is obtained by adding the axiom schema~\citep{Hosoi66}
\begin{equation}
p \vee \left(p \to q\right) \vee \neg q \label{eq:hosoi}
\end{equation}
to \I\ and it is characterized by frames of the form $(\lbrace 0,1\rbrace,\peq)$
where ${\peq}=\lbrace (0,0), (1,1), (0,1)\rbrace$.
Broadly speaking, the world $0$ (resp.\ $1$) refers to the world ``here'' (resp.\ ``there'').

We write $\Inf{\logic{X}}$ to specify the satisfaction relation in a concrete intermediate logic $\logic{X}$.
When the underlying logic is clear from the context, we omit $\logic{X}$ from the satisfaction relation.
The notions of satisfiability, validity and consistency are defined in an analogous way as with \I{}.
However, we prefix them with the logic when necessary.
For instance, when $\Gamma$ is consistent in the logic $\logic{X}$, we say that $\Gamma$ is $\logic{X}$-consistent.

Consistency in any intermediate logic $\logic{X}$ is equivalent to consistency in classical logic.
This is because, if a theory $\Gamma$ is classically satisfiable, it is trivially satisfiable in a one-world model;
conversely, if $\Gamma$ is $\logic{X}$-consistent,
there exists a finite model satisfying all formulas in $\Gamma$ at a world $x$.
Hence, a maximal world of the submodel generated by $x$ is a classical world that, in addition, satisfies $\Gamma$.
\begin{lemma}[\citealp{osnaar05a};~\citealp{dalen}]\label{lem:consistency}
  Let $\logic{X}$ and $\logic{Y}$ be two intermediate logics, and let $\Gamma$ be a theory.
  Then $\Gamma$ is $\logic{X}$-consistent iff $\Gamma$ is $\logic{Y}$-consistent.
\end{lemma}
\begin{definition}[Local semantic consequence]\label{def:sem-consequence}
  Let $\logic{X}$ be any intermediate logic and
  let $\Gamma$ and $\psi$ be a propositional theory and a propositional formula, respectively.
  We define $\psi$ as a \emph{local semantic consequence} of $\Gamma$,
  written $\Gamma \Inf{\logic{X}} \psi$,
  if,
  for each $\logic{X}$-model $\model = \tuple{\FName{F},V}$ and for each world $w \in \FName{F}$,
  \begin{align*}
    \big(\model, w \models \varphi\text{ for all } \varphi\in\Gamma\big) \text{ implies } \model, w \models \psi.
  \end{align*}

  When $\psi$ is replaced by a theory $\Delta$, we say that $\Delta$ is \emph{local semantic consequence} of $\Gamma$,
  written $\Gamma \Inf{\logic{X}} \Delta$,
  if for each model $\logic{X}$-model $\model = \tuple{\FName{F},V}$ and for each world $w \in \FName{F}$,
  \begin{align*}
  	\big(\model, w \models \varphi\text{ for all } \varphi\in\Gamma\big) \text{ implies } \left( \model, w \models \psi, \hbox{ for all } \psi \in \Delta\right).
   \end{align*}
\end{definition}
\begin{proposition}\label{prop:consequence:intermediate}
  Let $\logic{X} \subseteq \logic{Y}$ be two intermediate logics.
  For all theories $\Gamma$ and all formulas $\varphi$,
  $\Gamma \Inf{\logic{X}} \varphi$  implies $\Gamma \Inf{\logic{Y}} \varphi$.
\end{proposition}

We denote by
\(
\cn{\logic{X}}{\Gamma} := \lbrace \varphi \in \langP \mid \Gamma \Inf{\logic{X}} \varphi \rbrace
\)
the set of consequences obtained from $\Gamma$ within the intermediate logic $\logic{X}$.

\subsection{Bisimulations for Intuitionistic Logic}

In intuitionistic and modal logic,
a world $w$ in a model $\model$ cannot see the entire Kripke model,
it can only explore it \textit{locally}, step by step, by means of the different accessibility relations
($\peq$ in the case of \I).
The concept of bisimulation states that,
if two worlds $w$ and $w'$ are \textit{bisimilar},
it does not matter how one can try to explore from $w$ and $w'$ using modal formulas,
since they are behaviorally indistinguishable with respect to their logical properties.
Bisimulation in intuitionistic propositional logic was studied by~\cite{patterson97a}.
Formally, given two models $\model_1 = \tuple{(W_1,\peq_1),V_1}$ and $\model_2= \tuple{(W_2,\peq_2),V_2}$,
a bisimulation $\bisim$ is a relation on $W_1\times W_2$ that satisfies the following properties:
\begin{enumerate}[label=\textbf{C\arabic*}]
\item\label{cond:atoms} if $w_1 \bisim w_2$ then $V_1(w_1)= V_2(w_2)$
\item\label{cond:forth:imp} if $w_1 \bisim w_2$ then for all $v_1 \in W_1$, if $w_1 \peq_1 v_1$,
  then there exists $v_2\in W_2$ such that $w_2\peq_2v_2$ and $v_1 \bisim v_2$
\item\label{cond:back:imp} if $w_1 \bisim w_2$ then for all $v_2 \in W_2$, if $w_2 \peq_2 v_2$,
  then there exists $v_1\in W_1$ such that $w_1\peq_1v_1$ and $v_1 \bisim v_2$
\end{enumerate}
Condition~\ref{cond:atoms} ensures that two bisimilar worlds satisfy the same atoms.
Condition~\ref{cond:forth:imp} expresses that if we move forward from $w_1$ to a new world $v_1$,
we can find a matching move from $w_2$ to a world $v_2$ that is (logically) indistinguishable from $v_1$.
In other words, the second model can \textit{imitate} the movements of the first.
Condition~\ref{cond:back:imp} is the mirror image of~\ref{cond:forth:imp}:
now the second model makes a move and the first one must be able to imitate it.

Given two models $\model_1 = \tuple{(W_1,\peq_1),V_1}$ and $\model_2= \tuple{(W_2,\peq_2),V_2}$ with $w_1\in W_1$ and $w_2 \in W_2$,
we say that $\model_1$ and $\model_2$ are \emph{bisimilar},
if there exists a bisimulation $\bisim$ between $W_1$ and $W_2$ such that $w_1 \bisim w_2$.
The following lemma states that two bisimilar Kripke worlds satisfy the same formulas.
\begin{lemma}[\citealp{patterson97a}]\label{lem:bisim:il}
  Given two models $\model_1 = \tuple{(W_1,\peq_1),V_1}$ and $\model_2= \tuple{(W_2,\peq_2),V_2}$ and
  a bisimulation $\bisim$ on $W_1\times W_2$,
  then for all $w_1\in W_1$ and for all $w_2 \in W_2$,
  if $w_1\bisim w_2$ then for all $\varphi \in \langP$, $\model_1,w_1 \models \varphi$ iff $\model_2, w_2 \models \varphi$.
\end{lemma}
The proof of the lemma is done by structural induction and it uses~\ref{cond:atoms} to prove the case of the propositional variables while conditions~\ref{cond:forth:imp} and~\ref{cond:back:imp} are used to prove the case of implication.

\cite{osorio} show that the notion of $\logic{X}$-safe beliefs is independent of the intermediate logic $\logic{X}$;
replacing $\logic{X}$ with any proper intermediate logic leaves the set of safe beliefs unchanged.
Their approach, however, relies on a syntactic entailment relation and a Hilbert-style axiomatization of \I.
Here, we establish the same result semantically via bisimulations.
Our first lemma demonstrates that if a world $w$ in an intuitionistic model $\model$ satisfies all instances of the axiom $\neg p \vee \neg \neg p$,
then all maximal worlds in the subframe generated by $w$ can be merged into a single world.
\begin{proposition}\label{prop:weakexmiddle}
Let $\model = \tuple{(W,\peq), V}$ be an intuitionistic model and let $w \in W$	be such that $\model, w \models \lbrace \neg p \vee \neg \neg p\mid p \in \PV \rbrace$.
It follows that all maximal $\peq$-worlds in the subframe generated by $x$ satisfy the same propositional variables.
\end{proposition}
\begin{lemma}\label{lem:bisim:topwidth1}
Let $\model=\tuple{(W,\peq),V}$ an intuitionistic model and let $w \in W$ be such that $\model, w \models \lbrace \neg p \vee \neg \neg p\mid p \in \PV \rbrace$.
There exists a model $\model'=\tuple{(W',\peq'),V'}$,  $w'\in W'$  and a bisimulation $\bisim \subseteq W\times W'$ such that $w \bisim w'$ and the subframe generated by $w'$ has a unique maximal world with respect to $\peq'$.
\end{lemma}
\begin{proof}
From $\model, w \models \lbrace \neg p \vee \neg \neg p\mid p \in \PV \rbrace$ and Proposition~\ref{prop:weakexmiddle}, all maximal worlds in the subframe generated by $w$ satisfy the same propositional variables.
Let us define now the model $\model' := \tuple{(W',\peq'),V'}$ as follows:
\begin{itemize}
	\item $W' = \lbrace v\in W, u \mid w \peq v \hbox{ and } v \hbox{ is not } \peq\hbox{-maximal }\rbrace$, where $u\not \in W$ is a fresh world.
	\item $v \peq' v'$ if $v, v'\in W$ and $v\peq v'$ or $v' = u$; $u \peq' u$.
	\item $V'(u) := V(x)$, where $x\in W$, $w \peq x$ and $x$ is $\peq$-maximal; $V'(v):=V(v)$, for all $v\in W'$ with $v \not=u$.
\end{itemize}

	\begin{figure}[h!]\centering
		\begin{tikzpicture}
			[level distance=1.5cm,
			level 1/.style={sibling distance=1.5cm},
			level 2/.style={sibling distance=0.4cm}]
			\node (w1) {$w$}
			child[->] {node (v1) {$v_1$}
				child[->] {node (u1) {$u_1$} edge from parent node[left] {$\peq$}}
				child[->] {node (u2) {$u_2$} edge from parent node[right] {$\peq$}}
				edge from parent node[left,pos=.2] {$\peq$}
			}
			child[->] {node (v2) {$v_2$}
				child[->] {node (u3) {$u_3$} edge from parent node[left] {$\peq$}}
				edge from parent node[left,pos=.2] {$\peq$}
			}
			child[->] {node (v3) {$v_3$}
				child[->] {node (u4) {$u_4$} edge from parent node[left] {$\peq$}}
				child[->] {node (u5) {$u_5$} edge from parent node[right] {$\peq$}}
				edge from parent node[right,pos=.2] {$\peq$}
			};

			\node[right of = w1, node distance=6cm] (ww1) {$w$}
			child[->] {node (vv1) {$v_1$} edge from parent node[left,pos=.2] {$\peq'$}}
			child[->] {node (vv2) {$v_2$}
				child[->] {node (uu1) {$u$}edge from parent node[left,pos=.2] {$\peq'$}}
				edge from parent node[right,pos=.2] {$\peq'$}
			}
			child[->] {node (vv3) {$v_3$} edge from parent node[right,pos=.2] {$\peq'$}};

			\draw[->] (vv1) -- node[right,pos=.2] {$\peq'$} (uu1);
			\draw[->] (vv3) -- node[right,pos=.2] {$\peq'$} (uu1);

			\node[below of=u3,node distance=2cm] (label1) {$\model$};
			\node[below of=uu1,node distance=2cm] (label2) {$\model'$};

			\path[dashed, color=red]
			(w1) edge[] (ww1)
			(v1) edge[bend left] (vv1)
			(v2) edge[bend left] (vv2)
			(v3) edge[bend left] (vv3)
			(u1) edge[bend right] (uu1)
			(u2) edge[bend right] (uu1)
			(u3) edge[bend right] (uu1)
			(u4) edge[bend left] (uu1)
			(u5) edge[] (uu1)
			;

		\end{tikzpicture}
		\caption{A bisimulation relation, represented in red dashed lines, among two intuitionistic models $\model$ (on the left) and $\model'$ (with a unique maximal world, on the right).
			We assume that, $\model, w \models \lbrace \neg p \vee \neg \neg p\mid p \in \PV \rbrace$ and we define $V'$ as $V'(w):=V(w)$, $V'(v_i):= V(v_i)$ for all $i \in \lbrace 1,2,3\rbrace$ and
			$V'(u)$ can be set (for instance) to $V(u_1)$. Reflexivity and transitivity of $\peq$ and $\peq'$ is not represented for the sake of readability.}
		\label{fig:bisimilar:il}
	\end{figure}
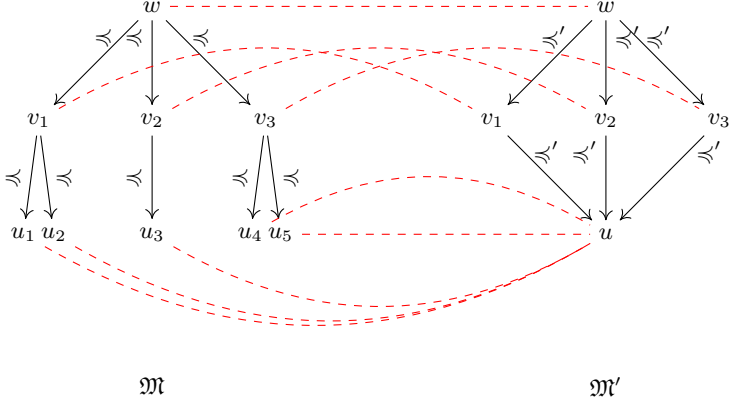

	It can be checked that $\model'$ is an intuitionistic model and, moreover, that there exists a relation $\bisim\subseteq W\times W'$, displayed in red dashed lines in Figure~\ref{fig:bisimilar:il}.
	In general, we map maximal worlds in $W$ to $u\in W'$ while the remaining worlds $v \in W$ that belong to the subframe generated by $w$ are mapped to themselves in $W'$ (where they also belong by construction). The reader can easily check that $\bisim$ is a bisimulation.
\end{proof}

The notion of bisimulation can also be used to contract intuitionistic models to \HT\ models.
In the following lemma, we identify the condition under which such contraction is possible.
Before presenting our result, we introduce the following notation.
\begin{definition}
	Let $\model = \tuple{(W,\peq),V}$ be an intuitionistic model and let $w\in W$.
	We define the sets
	\begin{align*}
		{\peq}(w)  &:= \lbrace v \in W \mid w \peq v\rbrace \\
		{\prec}(w) &:= \lbrace v \in W \mid w \peq v \hbox{ and } v \not = w\rbrace.
	\end{align*}
\end{definition}
Clearly, $w \in {\peq(w)}$ while $w \notin {\prec(w})$.
\begin{lemma}[Contraction lemma]\label{lem:contraction:il}
	Let $\model = \tuple{(W,\peq),V}$ be an intuitionistic model, let $T\subseteq \PV$ and let $w \in W$ satisfying the following conditions:
	\begin{enumerate}
		\item ${\peq}(w)$ has a unique maximal world, denoted by $u$, and
		\item $V(v)=T$, for all $v \in {\prec}(w)$.
	\end{enumerate}
	Then, there exists a \HT\ model $\model' = \tuple{(\lbrace 0,1 \rbrace,\peq'),V'}$ and a bisimulation $\bisim \subseteq W\times \lbrace 0,1\rbrace$ such that $w \bisim 0$
\end{lemma}
\begin{proof}
	Let us define $\model' = \tuple{(\lbrace 0,1\rbrace, \peq'),V'}$ as $V'(0) := V(w)$ and $V'(1):=V(u)$.
	Let us define the relation $\bisim \subseteq W\times \lbrace 0,1\rbrace$ as $\bisim := \lbrace ((w,0),(v,1))\mid v \in \prec(w)\rbrace$.
	It can be checked that $\bisim$ is a bisimulation among $\model$ and $\model'$.
	Figure~\ref{fig:ht:contraction} shows both $\model$ and $\model'$ together with $\bisim$  (in red dashed lines).
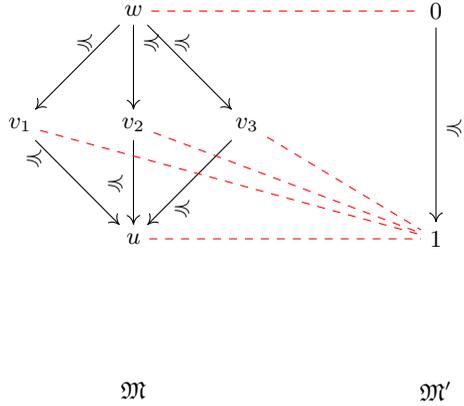
\begin{figure}[h!]\centering
		\begin{tikzpicture}
			[level distance=1.5cm,
			level 1/.style={sibling distance=1.5cm},
			level 2/.style={sibling distance=0.4cm}]

			\node[] (ww1) {$w$}
			child[->] {node (vv1) {$v_1$} edge from parent node[left,pos=.2] {$\peq$} }
			child[->] {node (vv2) {$v_2$} child {node (uu1) {$u$} edge from parent node[left] {$\peq$}} edge from parent node[right,pos=.2] {$\peq$}}
			child[->] {node (vv3) {$v_3$} edge from parent node[right,pos=.2] {$\peq$}};

			\draw[->] (vv1) -- node[left,pos=.2] {$\peq$} (uu1);
			\draw[->] (vv3) -- node[right,pos=.8] {$\peq$} (uu1);

			\node[right of = ww1, node distance=4cm] (www1) {$0$}
			child[->] {node[right of=uu1,node distance=4cm] (uuu1) {$1$}  edge from parent node[right] {$\peq'$}};

			\node[below of=uu1,node distance=2cm] (label2) {$\model$};
			\node[below of=uuu1,node distance=2cm] (label1) {$\model'$};

			\path[dashed, color=red]
			(ww1) edge[] (www1)
			(vv1) edge[] (uuu1)
			(vv2) edge[] (uuu1)
			(vv3) edge[] (uuu1)
			(uu1) edge[] (uuu1);

		\end{tikzpicture}
		\caption{An intuitionistic model $\model$, a \HT{} model $\model'$ and a bisimulations $\bisim$ (in red dashed lines) among them.
			As preconditions, $V(v) = T$ for all $v \in W$ with $v \not = w$, $V'(0) = V(w)$ and $V'(1) = V(u)$. Reflexivity and transitivity of $\peq$ and $\peq'$ are not represented for the sake of readability.}
			\label{fig:ht:contraction}
	\end{figure}
\end{proof}

\subsection{Equilibrium Logic}

The logic of here-and-there (\HT;~\citealp{heyting30a}) is fundamental in logic programming,
since it serves as the basis for equilibrium logic~(\EL;~\citealp{pearce96a,pearce06a}),
the most prominent logical characterization of \emph{stable models} and \emph{answer sets}~\citep{gellif88b}.
\EL\ extends \HT\ via a model selection criterion that induces nonmonotonicity.
For any two \HT\ models $\model_1 = (\FName{F},V_1)$ and $\model_2 = (\FName{F},V_2)$,
we define a partial order $\model_1\le \model_2$ holding when $V_1(1) = V_2(1)$ and $V_1(0)\subseteq V_2(0)$.
Strict inequality $\model_1 < \model_2$ holds if $\model_1 \le \model_2$ and $V_1 \neq V_2$.
A model $\model = (\FName{F},V)$ is called \emph{total} if $V(0) = V(1)$.

Equilibrium models are then defined as follows.
\begin{definition}[\citealp{pearce06a}]
  A total \HT\ model $\model = \tuple{\FName{F},V}$ is an equilibrium model of a formula $\varphi$ if
  \begin{enumerate}
  \item $\model, 0 \Inf{\HT} \varphi$, and
  \item $\model$ is $\le$-minimal, i.e., there is no $\model' < \model$ such that $\model', 0 \Inf{\HT} \varphi$.
  \end{enumerate}
  \qed
\end{definition}
 \section{Two Fixpoint Characterisations of Propositional Equilibrium Logic}\label{sec:fixpoint}

This section presents two characterizations of equilibrium models (and, consequently, answer sets) that
we extend to the temporal setting in Section~\ref{sec:main-contribution}.
The first characterization, originally defined by~\cite{pearce06a}, is based on the concept of \emph{theory completions},
which has also been used in autoepistemic and default logic~\citep{martru93a,besnard89}.

\subsection{Pearce's fixpoint characterization}

\begin{definition}[\citealp{pearce06a}]
  Let $\Gamma$ be a theory.
  A set $E$ of formulas extending $\Gamma$ is said to be a \emph{completion} of $\Gamma$ iff
  \begin{align*}
    E & = \cn{\HT}{\Gamma \cup \lbrace \neg \psi \mid \psi \not \in E \rbrace}.\tag*{\qed}
  \end{align*}
\end{definition}

Equilibrium models correspond precisely to completions in the propositional case.
For any model $\model$, we define
\begin{equation*}
\Th{\model}:= \lbrace \psi \mid   \model, 0 \Inf{\HT} \psi \rbrace.
\end{equation*}
The relation between equilibrium models and theory completions is made precise next.
\begin{proposition}[\citealp{pearce99b}]\label{prop:pearce}
For any theory $\Gamma$, there is a one-to-one correspondence between the equilibrium models of $\Gamma$ and the
completions of $\Gamma$.
In particular, $E = \Th{\model}$ for some equilibrium model $\model$ of $\Gamma$.
Similarly, any total \HT\ model $\model=\tuple{\FName{F}, V}$ is an equilibrium model of $\Gamma$ iff $V(0) =V(1)= E\cap\PV$
for some completion $E$ of $\Gamma$.
\end{proposition}

\subsection{\I-safe beliefs}\label{sec:safe-beliefs}

A slightly different fixpoint characterization of equilibrium logic can be given in terms of \I-\emph{safe beliefs}~\citep{osnaar05a},
which are typically defined in terms of entailment in intuitionistic and intermediate logics.
In this section, we provide a semantic reformulation of \I-safe beliefs,
offering a different perspective that is useful for our subsequent development of temporal safe beliefs.
\begin{definition}\label{def:safe-belief}
  A set $T$ of atoms is said to be an \I-\emph{safe belief set} of a theory $\Gamma$ if
  \begin{itemize}
  \item $\Gamma \cup \lbrace \neg \neg p \mid p \in T \rbrace \cup \lbrace \neg p \mid p \not \in T\rbrace$ is $\I$-consistent and
  \item $\Gamma \cup \lbrace \neg \neg p \mid p \in T \rbrace \cup \lbrace \neg p \mid p \not \in T\rbrace\Inf{\I} T$.
  \end{itemize}
\end{definition}

Analogous definitions can be reproduced for intermediate logics by replacing \I\ by any intermediate logic $\logic{X}$.
In fact, $\HT$-safe beliefs correspond to equilibrium models, as stated in the following lemma.
\begin{lemma}[\citealp{osnaar05a,pearce06a}]\label{lem:equivalence:prop}
  There is a one-to-one correspondence between both the $\HT$-safe beliefs sets of a theory $\Gamma$ and its equilibrium models. That is, for each \HT-safe belief set $T$,
  we can construct an equilibrium model $\model=\tuple{\FName{F}, V}$ where $V(0)=V(1)=T$. Conversely, given $\model$, we can extract the $\HT$-safe belief set $T=V(0)$.
\end{lemma}
\cite{osnaar05a} show that safe beliefs are independent of the chosen intermediate logic, as stated in lemmas~\ref{lem:il:safe-beliefs:auxiliary1} (below) and~\ref{prop:intermediate2}.
\begin{lemma}[\citealp{osnaar05a}]\label{lem:il:safe-beliefs:auxiliary1}
	Let $T$ be a set of atoms and let $\logic{X}$ and $\logic{Y}$ be two proper intermediate logics such that $\logic{X} \subseteq \logic{Y}$.
	For any propositional theory $\Gamma$, if $T$ is a $\logic{X}$-safe belief set of $\Gamma$, then $T$ is a $\logic{Y}$-safe belief set of $\Gamma$.
\end{lemma}
In \citep{osnaar05a}, Lemma~\ref{prop:intermediate2} is proved by using arguments from proof theory.
We provide a model-theoretic proof based on Lemma~\ref{lem:contraction:il}, which is easier to extend to temporal equilibrium logic.
\begin{lemma}\label{prop:intermediate2}
  Let $T$ be a set of atoms and let $\logic{X}$ and $\logic{Y}$ be two proper intermediate logics satisfying $\logic{X} \subseteq \logic{Y}$. For any propositional theory $\Gamma$,
  if $T$ is a $\logic{Y}$-safe belief set of $\Gamma$ then $T$ is a $\logic{X}$-safe belief set of $\Gamma$.
\end{lemma}
\begin{proof}
Let us assume that $T$ is a $\logic{Y}$-safe belief of $\Gamma$. It holds that
\begin{enumerate}[label=(\alph*)]
	\item $\Gamma \cup \lbrace \neg \neg p \mid p \in T \rbrace \cup \lbrace \neg p \mid p \not \in T\rbrace$ is $\logic{Y}$-consistent and
	\item $\Gamma \cup \lbrace \neg \neg p \mid p \in T \rbrace \cup \lbrace \neg p \mid p \not \in T\rbrace\Inf{\logic{Y}} T$.
\end{enumerate}
From
\begin{equation}
\Gamma \cup \lbrace \neg \neg p \mid p \in T \rbrace \cup \lbrace \neg p \mid p \not \in T\rbrace \label{eq:consistency}
\end{equation}
being $\logic{Y}$-consistent and Lemma~\ref{lem:consistency}, it follows that~\eqref{eq:consistency} is both $\logic{X}$-consistent and $\I$-consistent.
From the second item, the fact that $\logic{Y}\subseteq\HT$ and Proposition~\ref{prop:consequence:intermediate} it follows
\begin{equation}
\Gamma \cup \lbrace \neg \neg p \mid p \in T \rbrace \cup \lbrace \neg p \mid p \not \in T\rbrace\Inf{\HT} T. \label{ht:cons}
\end{equation}

Let $\model=\tuple{(W,\peq),V}$ be any intuitionistic model and $w \in W$ satisfying

\begin{enumerate}[label=(\alph*),start=3]
	\item\label{prop:intermediate2:item:a} $\model, w \Inf{\I} \Gamma$,
	\item\label{prop:intermediate2:item:b} $\model, w \Inf{\I}  \lbrace \neg \neg p \mid p \in T \rbrace $ and
	\item\label{prop:intermediate2:item:c} $\model, w \Inf{\I}  \lbrace \neg p \mid p \not \in T\rbrace $.
\end{enumerate}
Items~\ref{prop:intermediate2:item:b} and~\ref{prop:intermediate2:item:c} imply that $\model, w \Inf{\I} \lbrace \neg p \vee \neg \neg p\mid p \in \PV\rbrace$.
By Lemma~\ref{lem:bisim:topwidth1} there exists an intuitionistic $\model' = \tuple{(W',\peq'),V'}$, $w'\in W'$ and a bisimulation $\bisim\subseteq W\times W'$ such that $w\bisim w'$
and the subframe generated by $w'$ has an unique maximal world.
By Lemma~\ref{lem:bisim:il}, $\model',w' \Inf{\I} \Gamma$, $\model', w' \Inf{\I}  \lbrace \neg \neg p \mid p \in T \rbrace$ and $\model', w'\Inf{\I}   \lbrace \neg p \mid p \not \in T\rbrace$.

Let us denote by $u'\in W'$ the (unique) maximal world in the subframe generated by $w'$.
By the monotonicity property of intuitionistic logic, $\model',u'\Inf{\I}  \lbrace \neg \neg p \mid p \in T \rbrace$ and $\model', u'\Inf{\I}   \lbrace \neg p \mid p \not \in T\rbrace$.
Since $u'$ is maximal $V'(u') = T$.
In addition, we prove by induction on $\depth{((W',\peq'),v')}$, that for all $v' \in \peq'(w')$ (which includes the case $w'$ as well), $\model', v' \Inf{\I}T$.
\begin{enumerate}
	\item If $depth((W',v'))=1$ then $v'$ is maximal so $v'=u'$ (and we reason as in the base case).
	\item For the inductive step, let us assume that $\depth{((W',\peq'),v')} = n+1$ and the claim holds for every $x\in \prec'(v')$, that is, $\depth{((W',\peq'),x)} \le n$.
	\begin{enumerate}
		\item If $\prec'(v')=\emptyset$, $v'$ is maximal so $V'(v')=V'(u')=T$.
		\item If $\prec'(v')\not=\emptyset$ then, by induction hypothesis, $\model', x \models T$ for all $x\in\prec'(v')$.
		From $\model', w' \Inf{\I}  \lbrace \neg p \mid p \not \in T\rbrace $ and $w' \peq' v'$,
		$\model', x \not \models p$ for all $x\in\prec'(v')$ and all $p \in \PV\setminus T$.
		Therefore, $V'(x) = T$ for all for all $x\in\prec'(v')$.
		By Lemma~\ref{lem:contraction:il}, there exists a \HT{} model $\model'' = \tuple{(\lbrace 0,1\rbrace,\peq''),V''}$ and a bisimulation $\bisim' \subseteq W' \times \lbrace 0,1 \rbrace$ such that $v' \bisim' 0$.
		By Lemma~\ref{lem:bisim:il}, $\model'',0 \Inf{\HT}  \Gamma$, $\model'', 0 \Inf{\HT}   \lbrace \neg \neg p \mid p \in T \rbrace$ and $\model'', 0\Inf{\HT}   \lbrace \neg p \mid p \not \in T\rbrace$.
		Thank to~\eqref{ht:cons} it follows $\model'', 0 \Inf{\HT}  T$.
From $v' \bisim' 0$ and Lemma~\ref{lem:bisim:il}, $\model', v' \Inf{\I} T$.
	\end{enumerate}
\end{enumerate}
Therefore, $\model', w' \Inf{\I} T$.
Since $w \bisim w'$,  $\model, w \Inf{\I}T$.
Since $\model$ was chosen arbitrarily, it follows that
$\Gamma \cup \lbrace \neg \neg p \mid p \in T \rbrace \cup \lbrace \neg p \mid p \not \in T\rbrace\Inf{\I} T$.
By Proposition~\ref{prop:consequence:intermediate}, $\Gamma \cup \lbrace \neg \neg p \mid p \in T \rbrace \cup \lbrace \neg p \mid p \not \in T\rbrace\Inf{\logic{X}} T$.
Therefore, $T$ is a $\logic{X}$-safe belief of $\Gamma$.
\end{proof}
Note that the combination of lemmas~\ref{lem:il:safe-beliefs:auxiliary1} and~\ref{prop:intermediate2} allows us to replace \HT\ by any proper intermediate logic without altering the set of equilibrium models.
 \section{Temporal Intuitionistic and Intermediate Logics}\label{sec:itl}

Given a countable, possibly infinite set \PV\ of atoms, also called \emph{alphabet},
our temporal language $\langT$ consists of formulas generated by the following grammar:
\begin{displaymath}
  \varphi ::= \myatom\in\PV \mid \bot \mid \varphi\wedge\varphi \mid \varphi\vee\varphi \mid \varphi\to\varphi \mid
  \next\varphi \mid \varphi\until\varphi \mid \varphi\release\varphi
\end{displaymath}
This extends our basic language with temporal modal operators $\next$, $\until$, and $\release$.
The intended meaning of these operators is the following:
$\next \varphi$ means that $\varphi$ is true at the next time point.
$\varphi \until \psi$ means that $\varphi$ is true until $\psi$ is true.
For $\varphi \release \psi$ the meaning is not as direct as for the previous operators.
That is, $\varphi \release \psi$ means that $\psi$ is true until both $\varphi$ and $\psi$ become true simultaneously or
$\psi$ is true forever.
We also define several common derived operators like the Boolean connectives
\(
\top \eqdef \neg \bot
\),
\(
\neg \varphi \eqdef  \varphi \to \bot
\),
\(
\varphi \leftrightarrow \psi \eqdef (\varphi \to \psi) \wedge (\psi \to \varphi)
\),
and the unary temporal operators
$\alwaysF \varphi  \eqdef  \bot \release \varphi$ (always afterwards) and
$\eventuallyF \varphi\eqdef  \top \until \varphi$ (eventually afterwards).
A \emph{(temporal) theory} is a possibly infinite set of temporal formulas.

Formulas of $\langT$ are interpreted over intuitionistic temporal frames.
An \emph{intuitionistic temporal frame} is a tuple $\FName{D} = (W,\peq,\sfun)$,
where
$W$ is a non-empty set of (Kripke) worlds,
$\irel$ is a partial order,
and
$\sfun$ is a function from $W$ to $W$ satisfying the \emph{forward confluence} condition:
If $w \irel v $ then $ \msuccfct(w) \irel \msuccfct(v)$ for all $w, v \in W$.
Conversely, the \emph{backward confluence} condition stipulates that
if $\msuccfct(w) = v$ and $v\peq u$,
then there exists $t \in W$ such that $w \peq t$ and $\msuccfct(t)=u$ for all $w, v, u \in W$.
If $\FName{D}$ satisfies both confluence conditions,
we call $\FName{D}$ a \emph{persistent intuitionistic temporal frame}.
Figure~\ref{fig:confluence:forward} (resp.\ Figure~\ref{fig:confluence:backward}) shows
a graphical version of the forward (resp.\ backward) confluence condition.
\begin{figure}[h!]
	\begin{center}
		\begin{subfigure}{.5\columnwidth}\centering
			\begin{tikzpicture}[x=2cm,y=2cm]
				\node[] (bl) at (0,0) {$w$};
				\node (br) at (1,0) {$S(w)$};
				\node (tl) at (0,1) {$v$};
				\node (tr) at (1,1) {$S(v)$};
				\path[solid]
				(bl) edge[->] node [below] {$S$} (br)
				(tl) edge[->] node [above] {$S$} (tr)
				(bl) edge[->] node [above,sloped] {$\peq$} (tl)
				;
				\path[dashed]
				(br) edge[->] node [below,sloped] {$\peq$} (tr)
				;
			\end{tikzpicture}
			\caption{Forward confluence}
			\label{fig:confluence:forward}
		\end{subfigure}\begin{subfigure}{.5\columnwidth}\centering
			\begin{tikzpicture}[x=2cm,y=2cm]
				\node[] (bl) at (0,0) {$w$};
				\node (br) at (1,0) {$v$};
				\node (tl) at (0,1) {$t$};
				\node (tr) at (1,1) {$u$};

				\path[solid]
				(bl) edge[->] node [below] {$S$} (br)
				(br) edge[->] node [below,sloped] {$\peq$} (tr)
				;
				\path[dashed]
				(tl) edge[->] node [above] {$S$} (tr)
				(bl) edge[->] node [above,sloped] {$\peq$} (tl)
				;
			\end{tikzpicture}
			\caption{Backward confluence}
			\label{fig:confluence:backward}
		\end{subfigure}

	\end{center}
	\caption{Diagrams associated to forward and backward confluence.
          The above diagrams can always be completed if $S$ is forward or backward confluent
          (represented by means of dashed arrows).}\label{FigCO}
\end{figure}
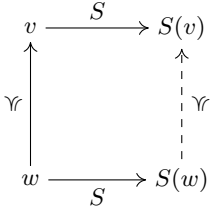
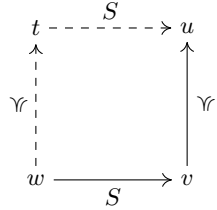

An \emph{intuitionistic temporal model}, or simply \emph{model}, is a tuple $\model=\tuple{\FName{D}, V}$ consisting of
an intuitionistic temporal frame $\FName{D}=(W,\peq,S)$ equipped with
a valuation function $\deffun V   W {\parts\PV} $ that is \emph{monotone} in the sense that
if $w \irel v $ then $ V(w) \subseteq V(v)$ for all $ w, v \in W$.
In the standard way, we define
$\msuccfct[^0](w) := w$ and
$\msuccfct[^{k+1}](w) := \msuccfct\left(\msuccfct[^{k}](w)\right)$ for $k \geq 0$.
Regarding the satisfaction relation,
the propositional connectives are satisfied as in \I\ (see Section~\ref{sec:il}).
The satisfaction of the temporal connectives is presented below.
\begin{enumerate}[start=6]
\item $\model, w \models \next \varphi $   iff $ \model, \msuccfct(w) \models \varphi $
\item  $\model, w \models \varphi \until \psi $ iff there exists $k \ge 0 $ such that $ \model, \msuccfct[^k](w) \models \psi$ and $\model, \msuccfct[^i](w) \models \varphi$ for all $\rangeco{i}{0}{k}$
\item 	$\model, w \models \varphi \release \psi $ iff for all $k \ge 0$, either $\model, \msuccfct[^k](w) \models \psi$
  or  $\model, \msuccfct[^i](w) \models \varphi$  for some $\rangeco{i}{0}{k}$.
\end{enumerate}

Figure~\ref{fig:iltl-frame} illustrates the satisfaction relation `$\models$'~\citep{BalbianiBDF20}.
Note that $\model, x\models \next p$ but $\model, x \not \models p$, while $\model, y \models p$ but
$\model, y \not \models \next p$.
From this, it follows that $\model, w \not \models (\next p\rightarrow p)\vee ( p \rightarrow \next p)$.
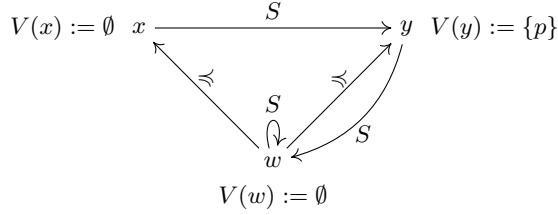
\begin{figure}[h!]\centering
	\begin{tikzpicture}[node distance=2.5cm]
		\node (a) [label=below:{$V(w) := \emptyset$}] {$w$};
		\node (b) [label=left:{$V(x) := \emptyset$}, above left of=a] {$x$};
		\node (c) [label=right:{$V(y) := \lbrace p \rbrace$}, above right of=a] {$y$};
		\path[->]
		(a) edge[] node[above] {$\peq$} (b)
		(a) edge[] node[above] {$\peq$} (c)
		;
		\path[->]

		(c) edge[bend left] node[below] {$S$} (a)
		(b) edge[] node[above] {$S$} (c)
		(a) edge[loop above] node[] {$S$} (a)
		;
	\end{tikzpicture}
	\caption{Example of an $\iltl$ model $\model = ((W,{\peq},S),V)$, where reflexivity and transitivity for $\peq$ are not represented.}
	\label{fig:iltl-frame}
\end{figure}
We refer to the intuitionistic temporal logic interpreted over the class of intuitionistic temporal frames as
\emph{intuitionistic temporal logic} (\iltl).
Formally, defined as
\begin{displaymath}
  \iltl := \lbrace \varphi \in \langT \mid \FName{D}\models \varphi \rbrace,
\end{displaymath}
where $\FName{D}$ is an intuitionistic temporal frame.
If in addition $\FName{D}$ is persistent, we denote by
\begin{displaymath}
  \itlb := \lbrace \varphi \in \langT \mid \FName{D}\models \varphi \rbrace
\end{displaymath}
the \emph{persistent intuitionistic temporal logic} (\itlb), i.e.,
the intuitionistic temporal logic interpreted over the class of the intuitionistic persistent frames.

The following proposition shows that $\iltl\not=\itlb$.
\begin{proposition}[\citealp{BalbianiBDF20}]\label{prop:valid:iltp}
The formulas $\left(\next p \to \next q\right) \to\next\left(p \to q\right)$ and $\left(\eventuallyF p \to \alwaysF q\right) \to\alwaysF\left(p \to q\right)$ are valid over the class of persistent intuitionistic temporal frames.
\end{proposition}
The formulas presented in the proposition above are valid in $\itlb$ but not in $\iltl$.	This leads to the following result.
\begin{corollary}[\citealp{BalbianiBDF20}] $\iltl \not= \itlb$. \end{corollary}

We remark that intuitionistic temporal frames impose minimal conditions on $\msuccfct$ and $\irel$ in order to
preserve the monotonicity of truth of formulas, in the sense that if $\model, w \Inf{\iltl} \varphi $ and $w\irel v$ then $\model, v \Inf{\iltl} \varphi$.
In the propositional case, the monotonicity property is guaranteed by the use of a monotone valuation.
In the temporal case, we additionally require the forward confluence property, which relates $\peq$ and $S$. Forward confluence ensures the satisfaction of temporal formulas is also monotone with respect to $\peq$.
The backward confluence property, while not required for monotonicity, allows us to show that maximal points are preserved under the temporal successor relation: if $w$ is a $\peq$-maximal point, then $S(w)$ is also maximal (see Proposition~\ref{prop:maximal:next}).
\begin{proposition}[monotonicity; \citealp{BalbianiBDF20}]\label{PropIntCond}
  Let $\model=\tuple{( W,\irel,\msuccfct),V}$ be an intuitionistic temporal model.
  For any $w,v \in W$,
  if $w \peq v$ then for any temporal formula $\varphi$, $\model, w \models \varphi$ implies $\model, v \models \varphi$.
\end{proposition}
Contrary to the non-temporal case (see Section~\ref{sec:il}),
consistency in $\iltl$ cannot be reduced to consistency in plain \LTL.
We provide here the counterexample proposed by~\cite{BalbianiBDF20}:
consider $\Gamma:=\lbrace \neg\next p, \neg\next\neg p \rbrace$.
In \LTL, this theory is equivalent to $\lbrace \next \neg p ,\next p\rbrace$ and it is inconsistent.
However, the \iltl\ model $\model$ shown in Figure~\ref{FigIMLA} satisfies $\Gamma$ at the world $w$
(in symbols, $\model, w \models \Gamma$).
Note that, in this case, $S$ is forward, but not backward, confluent.
Hence, the decidability of the satisfiability problem in \iltl\ is not a corollary of the \LTL\ case.
\begin{figure}[h!]
	\begin{center}
		\begin{tikzpicture}[node distance=2cm]
			\node (l) [label=below:{$V(w) := \emptyset$}]  {$w$};
			\node (r) [label=right:{$V(v) := \emptyset$},right of = l] {$v$};
			\node (u) [label=right:{$V(u) := \lbrace p \rbrace$},above of=r] {$u$};

			\draw[->] (l) edge  node[above]{$S$}(r);
			\draw[->,loop below] (r) edge node[below]{$S$} (r);
			\draw[->,loop left] (u) edge node[left]{$S$} (u);
			\draw[->] (r) edge  node[right] {$\peq$}(u);
			\end{tikzpicture}

	\end{center}
	\caption{Example of an \iltl\ model satisfying $\Gamma :=\lbrace  \neg\next p, \neg\next\neg p\rbrace$. Reflexivity and transitivity of $\peq$ are omitted for the sake of clarity.}
	\label{FigIMLA}
\end{figure}
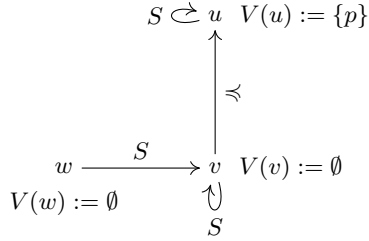
When adding the backward confluence property, $\Gamma$ becomes inconsistent.
However, adding such condition (i.e.\ replacing \iltl\ by \itlb) does not allow us to reduce \itlb-consistency to \LTL-consistency either.
To show this, we consider the theory $\Gamma := \lbrace \alwaysF \neg \neg p,  \neg \alwaysF p \rbrace$.
In $\LTL$, $\Gamma$ is equivalent to $\lbrace \alwaysF p, \neg \alwaysF p\rbrace$ and it is clearly inconsistent.
Now, consider the model $\model = \tuple{(\mathbb{N}\times \mathbb{N},\peq,S),V}$ shown in Figure~\ref{fig:itlb:tht:counter}.
In this case, the valuation $V$ is defined as $V((i,j)) := \lbrace p \rbrace$ if $ i < j$ and $\emptyset$ otherwise.
In fact, $\model$ is a $\itlb$ model since it possesses both forward and backward confluence properties.
Moreover,
\begin{enumerate}[label=\arabic*)]
\item For every $i \ge 0$ there exist $j > 0$ such that $\model, (i,j) \models p$,
  so $\model, (i,0)\models \neg \neg p$, for every $i \ge 0$.
  Consequently, $\model, (0,0) \models\alwaysF \neg \neg p $.
\item For every $j \ge 0$ there exist $i \ge 0$ such that $\model, (i,j) \not \models p$.
  Therefore, $\model, (0,j) \not \models \alwaysF p$, for every $j \ge 0$.
  Consequently, $\model, (0,0) \models \neg \alwaysF p$.
\end{enumerate}
As a consequence, $\M, (0,0) \models \Gamma$.
We remark that the depth of $\peq$ is not finite in $\model$.
\begin{figure}[h!]\centering
	\begin{tikzpicture}[node distance=1.5cm]
		\node (v00) []  {$(0,0)$};
		\node (v10) [right of = v00]  {$(1,0)$};
		\node (v20) [right of = v10]  {$(2,0)$};
		\node (v30) [right of = v20]  {$(3,0)$};
		\node (v40) [right of = v30]  {$(4,0)$};
		\node (v50) [right of = v40]  {$(5,0)$};
		\node (v60) [right of = v50]  {$(6,0)$};
		\node (v70) [right of = v60]  {$\cdots$};

		\node (v01) [below of = v00]  {$\mathbf{(0,1)}$};
		\node (v11) [below of = v10]  {$(1,1)$};
		\node (v21) [below of = v20]  {$(2,1)$};
		\node (v31) [below of = v30]  {$(3,1)$};
		\node (v41) [below of = v40]  {$(4,1)$};
		\node (v51) [below of = v50]  {$(5,1)$};
		\node (v61) [below of = v60]  {$(6,1)$};
		\node (v71) [below of = v70]  {$\cdots$};

		\node (v02) [below of = v01]  {$\mathbf{(0,2)}$};
		\node (v12) [below of = v11]  {$\mathbf{(1,2)}$};
		\node (v22) [below of = v21]  {$(2,2)$};
		\node (v32) [below of = v31]  {$(3,2)$};
		\node (v42) [below of = v41]  {$(4,2)$};
		\node (v52) [below of = v51]  {$(5,2)$};
		\node (v62) [below of = v61]  {$(6,2)$};
		\node (v72) [below of = v71]  {$\cdots$};

		\node (v03) [below of = v02]  {$\mathbf{(0,3)}$};
		\node (v13) [below of = v12]  {$\mathbf{(1,3)}$};
		\node (v23) [below of = v22]  {$\mathbf{(2,3)}$};
		\node (v33) [below of = v32]  {$(3,3)$};
		\node (v43) [below of = v42]  {$(4,3)$};
		\node (v53) [below of = v52]  {$(5,3)$};
		\node (v63) [below of = v62]  {$(6,3)$};
		\node (v73) [below of = v72]  {$\cdots$};

		\node (v04) [below of = v03]  {$\vdots$};
		\node (v14) [below of = v13]  {$\vdots$};
		\node (v24) [below of = v23]  {$\vdots$};
		\node (v34) [below of = v33]  {$\vdots$};
		\node (v44) [below of = v43]  {$\vdots$};
		\node (v54) [below of = v53]  {$\vdots$};
		\node (v64) [below of = v63]  {$\vdots$};

		\draw[->]
		(v00) edge  node[above]{$S$}(v10)
		(v10) edge  node[above]{$S$}(v20)
		(v20) edge  node[above]{$S$}(v30)
		(v30) edge  node[above]{$S$}(v40)
		(v40) edge  node[above]{$S$}(v50)
		(v50) edge  node[above]{$S$}(v60)
		(v60) edge  node[above]{$S$}(v70)

		(v01) edge  node[above]{$S$}(v11)
		(v11) edge  node[above]{$S$}(v21)
		(v21) edge  node[above]{$S$}(v31)
		(v31) edge  node[above]{$S$}(v41)
		(v41) edge  node[above]{$S$}(v51)
		(v51) edge  node[above]{$S$}(v61)
		(v61) edge  node[above]{$S$}(v71)

		(v02) edge  node[above]{$S$}(v12)
		(v12) edge  node[above]{$S$}(v22)
		(v22) edge  node[above]{$S$}(v32)
		(v32) edge  node[above]{$S$}(v42)
		(v42) edge  node[above]{$S$}(v52)
		(v52) edge  node[above]{$S$}(v62)
		(v62) edge  node[above]{$S$}(v72)

		(v03) edge  node[above]{$S$}(v13)
		(v13) edge  node[above]{$S$}(v23)
		(v23) edge  node[above]{$S$}(v33)
		(v33) edge  node[above]{$S$}(v43)
		(v43) edge  node[above]{$S$}(v53)
		(v53) edge  node[above]{$S$}(v63)
		(v63) edge  node[above]{$S$}(v73)
		;

		\draw[->]
		(v00) edge  node[left]{$\peq$}(v01)
		(v01) edge  node[left]{$\peq$}(v02)
		(v02) edge  node[left]{$\peq$}(v03)
		(v03) edge  node[left]{$\peq$}(v04)

		(v10) edge  node[left]{$\peq$}(v11)
		(v11) edge  node[left]{$\peq$}(v12)
		(v12) edge  node[left]{$\peq$}(v13)
		(v13) edge  node[left]{$\peq$}(v14)

		(v20) edge  node[left]{$\peq$}(v21)
		(v21) edge  node[left]{$\peq$}(v22)
		(v22) edge  node[left]{$\peq$}(v23)
		(v23) edge  node[left]{$\peq$}(v24)

		(v30) edge  node[left]{$\peq$}(v31)
		(v31) edge  node[left]{$\peq$}(v32)
		(v32) edge  node[left]{$\peq$}(v33)
		(v33) edge  node[left]{$\peq$}(v34)

		(v40) edge  node[left]{$\peq$}(v41)
		(v41) edge  node[left]{$\peq$}(v42)
		(v42) edge  node[left]{$\peq$}(v43)
		(v43) edge  node[left]{$\peq$}(v44)

		(v50) edge  node[left]{$\peq$}(v51)
		(v51) edge  node[left]{$\peq$}(v52)
		(v52) edge  node[left]{$\peq$}(v53)
		(v53) edge  node[left]{$\peq$}(v54)

		(v60) edge  node[left]{$\peq$}(v61)
		(v61) edge  node[left]{$\peq$}(v62)
		(v62) edge  node[left]{$\peq$}(v63)
		(v63) edge  node[left]{$\peq$}(v64)
		;

	\end{tikzpicture}
	\caption{A \itlb\ model $\model$ satisfying $\Gamma:=\lbrace \lbrace \alwaysF \neg \neg p, \neg \alwaysF p \rbrace \rbrace$ at $(0,0)$.
          The proposition $p$ is true in the worlds displayed boldface while false in those that are not.
          Reflexivity and transitivity of $\peq$ are not represented for the sake of readability.}
	\label{fig:itlb:tht:counter}
\end{figure}
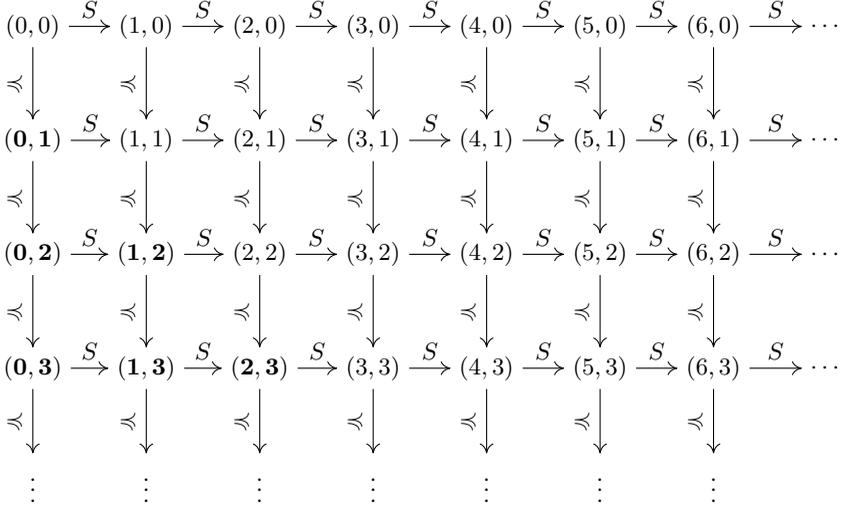

Such a counterexample leads to the following result.
\begin{proposition}
$\itlb$-consistency cannot be reduced to $\LTL$-consistency.
\end{proposition}

\subsection{Intermediate Temporal Logics}

As we have shown in the previous section,
intermediate temporal logics defined extending $\itlb$ with new axioms may not preserve consistency.
There is one extra condition that we need to impose on the intuitionistic temporal frames: \emph{the intuitionistic
  depth} ($\peq$-depth, for short) \emph{must be finite}.
Finite depth can be achieved by forcing \itlb\ to validate the schema $\bd{n}$, for some $n\ge 1$, as shown in
Lemma~\ref{depth:axiom}.
We define the family $\itlbd{n}$ of intuitionistic temporal logics as
\begin{equation*}
\itlbd{n} := \lbrace  \varphi\in \langT \mid  \FName{D} \models \varphi \rbrace,
\end{equation*}
where $\FName{D}=(W,\peq,S)$ is a $\itlb$ frame with $\depth{(W,\peq)}\le n$.
\begin{observation}\label{observation:ltl}
  $\itlbd{1}$ corresponds to \LTL, that is, $\itlbd{1}$ frames $\FName{D}=(W,\peq,S)$ such that $\depth{(W,\peq)}= 1$.\footnote{Note that the axiom schema $\bd{1}$ corresponds to the excluded middle axiom, which is added to \I\ to get classical logic and to \iltl\ to obtain \LTL.}
\end{observation}
\begin{proposition}\label{prop:maximal:next}
	For all \itlbd{n} frame $\FName{D}=(W,\peq,S)$ and for all $w \in W$, if $w$ is maximal w.r.t.\ $\peq$ then $S(w)$ is maximal.
\end{proposition}
Note that the proposition above can be proved only if the backward-confluence property holds.
In addition, using the result above, we can prove that $\itlbd{n}$-consistency can be reduced to \LTL-consistency.
\begin{lemma}\label{lem:itlbn2ltl-consistency}
Any temporal theory $\Gamma$ is $\itlbd{n}$-consistent iff $\Gamma$ is \LTL-consistent.
\end{lemma}

To the best of our knowledge, the family of intermediate temporal logics has not been defined in the literature.
We define those logics as extensions of \itlbd{n}.\footnote{In other words, we want $\itlbd{n}$ to play the same role as \I\ does in the propositional case.}
\begin{definition}[Intermediate temporal logic]
  An \emph{intermediate temporal logic} in the language \langT\ is any set of formulas $\logic{X}$
  satisfying the following conditions:
  \begin{enumerate}
  \item $\itlbd{n} \subseteq \logic{X} \subseteq \LTL$
  \item $\logic{X}$ is closed under \emph{modus ponens},
    i.e., $\varphi, \varphi \to \psi \in \logic{X}$ implies $\psi \in \logic{X}$
  \item $\logic{X}$ is closed under \emph{necessitation},
  i.e., $\psi \in \logic{X}$ implies $\next \psi \in \logic{X}$ and $\alwaysF \psi \in \logic{X}$
  \item $\logic{X}$ is closed under \emph{uniform substitution},
    i.e., $\varphi \in \logic{X}$ implies $\varphi\mathbf{s} \in \logic{X}$ for any $\varphi \in \langT$ and a substitution $\mathbf{s}$
  \end{enumerate}
\end{definition}
As in the propositional case,
an intermediate temporal logic is said to be \emph{proper}
if it is different from $\LTL$.
Assuming that any intermediate temporal logic extends \itlbd{n},
the following result directly follows from Lemma~\ref{lem:itlbn2ltl-consistency}.
\begin{corollary}\label{cor:intermediate2ltl:consistency}
  Let $\logic{X}$ be any intermediate temporal logic.
  Any temporal theory $\Gamma$ is $\logic{X}$-consistent iff $\Gamma$ is \LTL-{consistent}.
\end{corollary}
The following proposition shows that the depth of a generated (intuitionistic) subframe does not increase between two worlds $x$ and $S(x)$.
\begin{proposition}\label{prop:depth:next}
  For any intuitionistic temporal frame $\FName{D} = (W,\peq,S)$ and for an $w\in W$ and for any $n\ge1 $,
  if $\depth{((W,\peq),w)}\le n$ then $\depth{((W,\peq),S(w))}\le n$.
\end{proposition}

Similar to the propositional case, intermediate temporal logics defined as extensions of $\itlbd{n}$ satisfy the following proposition.
\begin{proposition}\label{prop:consistency:itl}
Let $\logic{X}$ and $\logic{Y}$ be two intermediate temporal logics satisfying $\logic{X} \subseteq \logic{Y}$. A theory $\Gamma$ is $\logic{X}$-consistent iff $\Gamma$ is $\logic{Y}$-consistent.
\end{proposition}
Furthermore, semantic entailment is preserved when strengthening the logic.
\begin{proposition}\label{prop:entailment:itl}
 Let $\logic{X}$ and $\logic{Y}$ be two intermediate temporal logics satisfying $\logic{X} \subseteq \logic{Y}$.
 For any theories $\Gamma$ and $\Delta$,
 \begin{displaymath}
 	\Gamma \Inf{\logic{X}} \Delta \hbox{ implies } \Gamma \Inf{\logic{Y}} \Delta.
 \end{displaymath}
\end{proposition}

The strongest proper intermediate temporal logic is the logic of here-and-there (\THT; \citealp{baldie16}).
In our setting, this logic is obtained by adding Axiom~\eqref{eq:hosoi} to $\itlbd{2}$.\footnote{It can also be obtained by adding Axiom~\eqref{eq:hosoi} to \itlb. However, in this paper,
  we take \itlbd{n} as the root of our constructions.}
\THT\ frames are of the form $(W,\peq,S)$ where
\begin{eqnarray*}
	W &=& \mathbb{N}\times \lbrace 0,1\rbrace,\\
	\peq &=& \lbrace ((i,h),(i,t)) \mid ((i,h),(i,t)) \in W \times W \hbox{ such that } h \le t \rbrace  \hbox{ and }\\
	S &=& \lbrace ((i,k),(i+1,k)) \mid ((i,k),(i+1,k)) \in W \times W\rbrace.
\end{eqnarray*}
In the definition above,
pairs of the form $(i,0)$ (resp.\ $(i,1)$) represent the ``here'' (resp.\ ``there'') world at each time point $i$.

\subsection{Bisimulations for Intuitionistic Temporal Logics}

The notion of intuitionistic temporal bisimulation was introduced in~\citep{BalbianiBDF20} and we use it in what follows
to contract intuitionistic temporal models into \THT.
Given two intuitionistic temporal models $\model_1=\tuple{(W_1,\peq_1,S_1),V_1}$ and $\model_2=\tuple{(W_2,\peq_2,S_2),V_2}$,
a relation $\bisim \subseteq W_1\times W_2$ is an intuitionistic temporal bisimulation,
if it satisfies Conditions~\ref{cond:atoms}-\ref{cond:back:imp} together with the following ones:
\begin{enumerate}[label=\textbf{C\arabic*},start=5]
\item \label{cond:next} If $w_1 \bisim w_2$ then $S(w_1) \bisim S(w_2)$;
\item \label{cond:until:forth} If $w_1 \bisim w_2$ then for all $k_1\ge 0$ there exists $k_2\ge 0$ and $(v_1,v_2) \in W_1 \times W_2$ such that
  \begin{enumerate}
  \item $v_2 \peq S^{k_2}(w_2)$, $S^{k_1}(w_1)\peq v_1$ and $v_1 \bisim v_2$ and
  \item for all $\rangeco{j_2}{0}{k_2}$ there exists $\rangeco{j_1}{0}{k_1}$ and $(u_1,u_2)\in W_1\times W_2$ such that $S^{j_1}(w_1) \peq u_1$, $u_2 \peq S^{j_2}(w_2)$ and $u_1 \bisim u_2$.
  \end{enumerate}
\item \label{cond:until:back} If $w_1 \bisim w_2$ then for all $k_2\ge 0$ there exists $k_1\ge 0$ and $(v_1,v_2) \in W_1 \times W_2$ such that
  \begin{enumerate}
  \item $v_1 \peq S^{k_1}(w_1)$, $S^{k_2}(w_2)\peq v_2$ and $v_1 \bisim v_2$ and
  \item for all $\rangeco{j_1}{0}{k_1}$ there exists $\rangeco{j_2}{0}{k_2}$ and $(u_1,u_2)\in W_1\times W_2$ such that $S^{j_2}(w_2)\peq u_2$, $u_1 \peq S^{j_1}(w_1)$ and $u_1 \bisim u_2$.
  \end{enumerate}
\item \label{cond:release:forth} If $w_1 \bisim w_2$ then for all $k_2\ge 0$ there exists $k_1\ge 0$ and $(v_1,v_2) \in W_1 \times W_2$ such that
  \begin{enumerate}
  \item $v_2 \peq S^{k_2}(w_2)$, $S^{k_1}(w_1)\peq v_1$ and $v_1 \bisim v_2$ and
  \item for all $\rangeco{j_1}{0}{k_1}$ there exists $\rangeco{j_2}{0}{k_2}$ and $(u_1,u_2)\in W_1\times W_2$ such that $S^{j_1}(w_1) \peq u_1$, $u_2 \peq S^{j_2}(w_2)$ and $u_1 \bisim u_2$.
  \end{enumerate}
\item \label{cond:release:back} If $w_1 \bisim w_2$ then for all $k_1\ge 0$ there exists $k_2\ge 0$ and $(v_1,v_2) \in W_1 \times W_2$ such that
  \begin{enumerate}
  \item $v_1 \peq S^{k_1}(w_1)$, $S^{k_2}(w_2)\peq v_2$ and $v_1 \bisim v_2$ and
  \item for all $\rangeco{j_2}{0}{k_2}$ there exists $\rangeco{j_1}{0}{k_1}$ and $(u_1,u_2)\in W_1\times W_2$ such that  $S^{j_2}(w_2) \peq u_2$, $u_1 \peq S^{j_1}(w_1)$,  and $u_1 \bisim u_2$.
  \end{enumerate}
\end{enumerate}

Conditions~\ref{cond:next}-\ref{cond:release:back} play the same role as
conditions~\ref{cond:forth:imp}-~\ref{cond:back:imp} in the case of \I\ but they affect the temporal modalities.
More precisely,
Conditions~\ref{cond:until:forth} and~\ref{cond:until:back} (resp.\ conditions~\ref{cond:release:forth} and~\ref{cond:release:back}) simulate the behavior of the until (resp.\ release) operator.
Note that, due confluence of both $\peq$ and $S$, the forth and back conditions for the binary temporal modalities involve both relations.
Condition~\ref{cond:next}, which is used to simulate the next modality, is not divided into two conditions because the next operator is interpreted in terms of a function.

Given two intuitionistic temporal models
$\model_1 = \tuple{(W_1,\peq_1, S_1),V_1}$ and $\model_2= \tuple{(W_2,\peq_2, S_2),V_2}$ with $w_1\in W_1$ and $w_2 \in W_2$,
we say that $\model_1$ and $\model_2$ are bisimilar,
if there exists an intuitionistic temporal bisimulation $\bisim$ between $W_1$ and $W_2$ such that $w_1 \bisim w_2$.
The following lemma states that two bisimilar Kripke worlds satisfy the same temporal formulas.
\begin{lemma}[\citealp{BalbianiBDF20}]\label{lem:bisim:itl}
  Given two models
  $\model_1 = \tuple{(W_1,\peq_1,S_1),V_1}$ and
  $\model_2= \tuple{(W_2,\peq_2,S_2),V_2}$ and
  a bisimulation $\bisim$ on $W_1\times W_2$,
  we have for all $w_1\in W_1$ and for all $w_2 \in W_2$,
  if $w_1\bisim w_2$ then for all $\varphi \in \langT$, $\model_1,w_1 \models \varphi$ iff $\model_2, w_2 \models \varphi$.
\end{lemma}
In the proof of the previous lemma, Condition~\ref{cond:atoms} is used to establish the case of propositional variables.
Conditions~\ref{cond:back:imp} and~\ref{cond:forth:imp} are used to handle the case of implication.
Conditions~\ref{cond:until:back} and~\ref{cond:until:forth} are employed for the $\until$ operator, while
Conditions~\ref{cond:release:back} and~\ref{cond:release:forth} are used for the $\release$ operator.
Finally, Condition~\ref{cond:next} is used in the proof of the case of the $\next$ connective.

In~\citep{BalbianiBDF20}, where \iltl\ and \itlb\ are studied in detail,
the authors proved that both satisfiability (resp.\ validity) on arbitrary \iltl\ models is equivalent to satisfiability (resp.\ validity) on the so-called \emph{expanding models}, defined in the following theorem.

\begin{theorem}[\citealp{BalbianiBDF20,BoudouDF17}]\label{thm:itl}
Every \iltl\ model $\model$ can be unfolded into an expanding model $\model' = \tuple{(W',\peq',S'),V'}$,
which satisfies the following properties:
\begin{enumerate}
	\item For all $i \ge 0$ and $w'\in W'$, the intuitionistic subframe generated by each $S'^i(w')$ is a tree,
	\item the sequence of trees induced by $S'^i(w')$ is a sequence of disjoint trees, and
	\item for all $v' \in W'$ and for all $i\ge 0$, if $S'^i(w') \peq v'$ then $S'^{i+1}(w') \peq S'(v')$, that is, $S'(v')$ falls in the tree generated by $S'^{i+1}(w')$.
\end{enumerate}
\end{theorem}

We refer the reader to~\citep{BalbianiBDF20} for more details about such unfolding.
Since $\itlbd{n}$ is contained in $\iltl$ then the construction above can be also applied to models of finite $\peq$-depth.
From now on, when displaying the \itlbd{n} models, we consider expanding models.

Bisimulations for intuitionistic temporal logics allow us to extend lemmas~\ref{lem:bisim:topwidth1} and~\ref{lem:contraction:il} to the temporal case.

\begin{lemma}\label{lem:bisim:itl:topwidth1}
	Let $\model=\tuple{(W,\peq,S),V}$ be an \itlbd{n} model and let $w \in W$ be such that $\model, w \models \lbrace \alwaysF(\neg p \vee \neg \neg p) \mid p \in \PV \rbrace$.
	Then, there exists an \itlbd{n} model $\model'=\tuple{(W',\peq',S'),V'}$, $w'\in W'$ and a bisimulation $\bisim\subseteq W \times W'$ such that, for all $i\ge 0$ both $S^i(w) \bisim S'^i(w')$ and each (intuitionistic) subframe generated by each $S'^i(w')$ has an unique maximal world, $u_i$, w.r.t. $\peq'$.
\end{lemma}

\begin{proof} Since $\model, w \models   \lbrace \alwaysF(\neg p \vee \neg \neg p) \mid p \in \PV \rbrace$ implies $\model, S^i(w) \models \lbrace \neg p \vee \neg \neg p \mid p \in \PV \rbrace$, for all $i \ge 0$.
Since $\model$ is an \itlbd{n} model, the intuitionistic depth of $(W,\peq)$ is finite and,
because of Proposition~\ref{prop:weakexmiddle},
all maximal worlds in $\peq(S^i(w))$
satisfy the same propositional variables.

Let us define the \itlbd{n} model $\model' := \tuple{(W',\peq',S'),V'}$, where its corresponding frame is defined as:
		\begin{align*}
			W' & :=  \lbrace v, u_i \mid i \ge 0 \hbox{, } v \hbox{ is } \peq-\hbox{maximal in} \peq(S^i(w)) \hbox{ and } u_i \not \in W \hbox{ is a fresh world}\rbrace; \\
\peq' &: =  \lbrace (u_i,u_i), (v,u_i) \mid i \ge 0,\; v, u_i \in W' \hbox{ and } S^i(w) \peq v \rbrace\\
		 	   & \cup \lbrace (v,v') \mid v,v' \in W'\hbox{ and }  v \peq v'\rbrace;\\
			S' & :=  \lbrace (u_i,u_{i+1}), (v,u_{i+1}) \mid i \ge 0,\; S^i(w)\peq v \hbox{ and } S(v)=x \hbox{, with } x \peq-\hbox{maximal}\rbrace\\
			   & \cup  \lbrace (v,v') \mid (v,v')\in W' \times W' \hbox{ and }  v S v' \rbrace. \\
\end{align*}
We verify that it $(W',\peq',S')$ is an \itlbd{n} frame.
Since $\peq$ is a partial order relation,  by construction, $\peq'$ is.
Since $S$ is a function, it can be verified that $S'$ is a function as well.
Also by construction, $(W',\peq')$ is of finite depth.
We readily check that $S'$ and $\peq'$ are forward and backward confluent. For the forward confluence let us consider
$v,v'\in W$ satisfying $v\peq'v'$.
In order to prove that $S'(v) \peq' S'(v')$ we need to consider several cases.

\begin{enumerate}
	\item If $v = u_i$, for some $i \ge 0$, by construction, $v'=u_i$ too.
	Therefore, $u_{i+1} = S'(v) \peq' S'(v') = u_{i+1}$.
	\item If $v \not = u_i$ but  $v'=u_i$ then $S'(v') = u_{i+1}$ and $v \in W$.
	If $S(v)$ is maximal then $u_{i+1}=S'(v)\peq S'(v') =u_{i+1}$ by definition.
	If not, take $S'(v) := S(v)\in W'$. By construction, $S'(v)\peq u_{i+1}$.
	\item If $v \not = u_i$ and  $v'\not =u_i$ it follows that $v \peq v'$.
	Since $\peq$ and $S$ are forward confluent then $S(v) \peq S(v')$.
	If either $S(v)$ or $S(v')$ are maximal w.r.t. $\peq$, we follow a similar reasoning as in the previous two items in order to check $S'(v) \peq' S'(v')$.
	If neither $S(v)$ nor $S(v')$ are maximal then $S'(v)=S(v)$ and $S'(v') = S(v')$ so $S'(v) \peq' S'(v')$ by construction.
\end{enumerate}

In any case we conclude that, if $v \peq' v'$ then $S'(v) \peq S'(v')$ as requested.

For the backward confluence, let us take three arbitrary worlds $v, y,z \in W'$ satisfying $S'(v) = y \peq' z$.
We show that there exists $t \in  W'$ such $v \peq' t$ and $S'(t) = z$.
As for the forward case, we proceed by cases.

\begin{enumerate}
	\item If $ y = u_{i+1}$ then $z = u_{i+1}$ by construction.
	Since $S'(v)= u_{i+1}$ then take $t:=u_i$. By definition, $v\peq' u_i$ and $S'(u_i) = u_{i+1} = z$.
	\item If $ y \not= u_{i+1}$ and $z=u_{i+1}$ then take $t:= u_i$.
	$v\peq' t$ by construction and $S'(u_i) = u_{i+1} = z$.
	\item  If $ y \not= u_{i+1}$ and $z\not =u_{i+1}$ then $y,z \in W$, $y \peq z$ and $z$ is not maximal w.r.t. $\peq$.
			Since $y \not= u_{i+1}$ then $S'(v) = y\not = u_{i+1}$.
			Therefore, $v\in W$.
			By construction, $S(v) = y$.
			Because of the backward confluence property, there exist $x \in W$ such that
			$v\peq x$ and $S(x) = z$.
			The world $x$ cannot be maximal, otherwise $z$ would be maximal in view of Proposition~\ref{prop:maximal:next} and it would not belong to $W'$.
			By construction, $x \in W'$ and $S'(x) = z$.
			Set $t:=x$ so we would get $v \peq' x$ and $S'(x)= z$.
\end{enumerate}

As a consequence, $(W',\peq',S)$ is an \itlbd{n} frame.

	\begin{figure}[h!]\centering

		\resizebox{\textwidth}{!}{
			\begin{tikzpicture}
				[level distance=1.5cm,
				level 1/.style={sibling distance=1.5cm},
				level 2/.style={sibling distance=0.8cm}]
				\node (w1) {$w$}
				child[->] {node (v1) {$v_{0,0}$}
					child[->] {node (u1) {$u_{0,0}$} edge from parent node[right] {$\peq$}}
					child[->] {node (u2) {$u_{0,1}$} edge from parent node[right] {$\peq$}}
					edge from parent node[left,pos=.2] {$\peq$}
				}
				child[->] {node (v2) {$v_{0,1}$}
					child[->] {node (u3) {$u_{0,2}$} edge from parent node[left] {$\peq$}}
					edge from parent node[right,pos=.2] {$\peq$}
				};

				\node[right of=w1, node distance=3.5cm] (ww1) {$S^1(w)$}
				child[->] {node (vv1) {$v_{1,0}$}
					child[->] {node (uu1) {$u_{1,0}$} edge from parent node[left] {$\peq$}}
					edge from parent node[left,pos=.2] {$\peq$}
				}
				child[->] {node (vv2) {$v_{1,1}$}
					child[->] {node (uu2) {$u_{1,2}$} edge from parent node[right] {$\peq$}}
					edge from parent node[right,pos=.2] {$\peq$}
				};

				\node[right of=ww1, node distance=3.5cm] (www1) {$S^2(w)$}
				child[->] {node (vvv1) {$v_{2,0}$}
					child[->] {node (uuu1) {$u_{2,0}$} edge from parent node[left] {$\peq$}}
					edge from parent node[left,pos=.2] {$\peq$}
				}
				child[->] {node (vvv2) {$v_{2,1}$}
					child[->] {node (uuu2) {$u_{2,1}$} edge from parent node[right] {$\peq$}}
					edge from parent node[right,pos=.2] {$\peq$}
				};

				\node[right of=www1, node distance=3.5cm] (wwww1) {$\cdots$}
				child[->] {node (vvvv1) {$\cdots$} edge from parent[draw=none]
					child[->] {node (uuuu1) {$\cdots$} edge from parent[draw=none]}
				};

				\path[->]
				(w1)   edge  node[above] {$S$} (ww1)
				(ww1)  edge  node[above] {$S$} (www1)
				(www1) edge  node[above] {$S$} (wwww1)

				(v1)  edge[bend left]  node[above] {$S$} (vv1)
				(v2)  edge[bend left]  node[above] {$S$} (vv2)

				(vv1)  edge[bend left]  node[above] {$S$} (vvv2)
				(vv2)  edge[bend left]  node[above] {$S$} (vvv1)

				(vvv1)  edge[bend left]  node[above] {$S$} (vvvv1)
				(vvv2)  edge[bend left]  node[above] {$S$} (vvvv1)

				(u1)  edge[bend right]  node[below] {$S$} (uu1)
				(u2)  edge[bend right]  node[below] {$S$} (uu2)
				(u3)  edge[]  node[above] {$S$} (uu1)

				(uu1)  edge[bend right]  node[below] {$S$} (uuu1)
				(uu2)  edge[bend right]  node[below] {$S$} (uuu2)

				(uuu1)  edge[bend right]  node[below] {$S$} (uuuu1)
				(uuu2)  edge[bend right]  node[below] {$S$} (uuuu1)
				;

				\node[below of=w1, node distance=6cm] (wp1) {$w$}
				child[->,sibling distance=.8cm] { node (vp1) {$v_{0,0}$}  edge from parent node[left,pos=.2] {$\peq'$}}
				child[->]{ child {node (up1) {$u_0$} edge from parent[draw=none]}edge from parent[draw=none]}
				child[->,sibling distance=.8cm] { node (vp2) {$v_{0,1}$}  edge from parent node[right,pos=.2] {$\peq'$}};
				\draw[->] (vp1) --  node[left,pos=.8] {$\peq'$}(up1);
				\draw[->] (vp2) --  node[right,pos=.8] {$\peq'$}(up1);

				\node[right of=wp1, node distance=3.5cm] (wwp1) {$S^1(w)$}
				child[->,sibling distance=.8cm] { node (vvp1) {$v_{1,0}$}
					edge from parent node[left,pos=.2] {$\peq'$}
				}
				child[->]{ child {node (uup1) {$u_1$} edge from parent[draw=none]}edge from parent[draw=none]}
				child [->,sibling distance=.8cm] {node (vvp2) {$v_{1,1}$}
					edge from parent node[right,pos=.2] {$\peq'$}
				};
				\draw[->] (vvp1) -- node[left,pos=.8] {$\peq'$}(uup1);
				\draw[->] (vvp2) -- node[right,pos=.8] {$\peq'$}(uup1);

				\node[right of=wwp1, node distance=3.5cm] (wwwp1) {$S^1(w)$}
				child[->,sibling distance=.8cm] {node (vvvp1) {$v_{2,0}$}
					edge from parent node[left,pos=.2] {$\peq'$}
				}
				child[->]{ child[->] {node (uuup1) {$u_2$} edge from parent[draw=none]}edge from parent[draw=none]}
				child[->,sibling distance=.8cm] {node (vvvp2) {$v_{2,1}$}
					edge from parent node[right,pos=.2] {$\peq'$}
				};
				\draw[->] (vvvp1) -- node[left,pos=.8] {$\peq'$}(uuup1);
				\draw[->] (vvvp2) -- node[right,pos=.8] {$\peq'$}(uuup1);

				\node[right of=wwwp1, node distance=3.5cm] (wwwwp1) {$\cdots$}
				child {node (vvvvp1) {$\cdots$} edge from parent[draw=none]
					child {node (uuuup1) {$\cdots$} edge from parent[draw=none]}
				};

				\node[right of=vvvv1]{$\model$};
				\node[right of=vvvvp1]{$\model'$};

				\path[->]
				(wp1) edge[] node[above]{$S'$}(wwp1)
				(wwp1) edge[] node[above]{$S'$}(wwwp1)
				(wwwp1) edge[] node[above]{$S'$}(wwwwp1)

				(vp1) edge[bend left] node[above]{$S'$}(vvp1)
				(vvp1) edge[bend left] node[above]{$S'$}(vvvp1)
				(vvvp1) edge[bend left] node[above]{$S'$}(vvvvp1)

				(vp2) edge[bend right] node[above]{$S'$}(vvp2)
				(vvp2) edge[bend right] node[above]{$S'$}(vvvp2)
				(vvvp2) edge[bend right] node[above]{$S'$}(vvvvp1)

				(up1) edge[] node[above]{$S'$}(uup1)
				(uup1) edge[] node[above]{$S'$}(uuup1)
				(uuup1) edge[] node[above]{$S'$}(uuuup1)
				;

				\path[dashed,color=red]
				(w1) edge[] (wp1)
				(ww1) edge[] (wwp1)
				(www1) edge[] (wwwp1)

				(v1) edge[bend right] (vp1)
				(vv1) edge[bend right] (vvp1)
				(vvv1) edge[bend right] (vvvp1)

				(v2) edge[bend left] (vp2)
				(vv2) edge[bend left] (vvp2)
				(vvv2) edge[bend left] (vvvp2)

				(u1) edge[] (up1)
				(u2) edge[] (up1)
				(u3) edge[] (up1)
				(uu1) edge[] (uup1)
				(uu2) edge[] (uup1)
				(uuu1) edge[] (uuup1)
				(uuu2) edge[] (uuup1)
				;
			\end{tikzpicture}
		}
		\caption{Two $\itlbd{n}$ models $\model=\tuple{(W,\peq,S),V}$ and $\model'=\tuple{(W',\peq',S'),V'}$. Under the assumption that $\model, w \models \lbrace \alwaysF(\neg p \vee \neg \neg p) \mid p \in \PV \rbrace$,
		for all $i\ge 0$, all maximal worlds in $\peq(S^i(w))$ satisfy the same atoms.
		By setting $V'(v):=V(v)$ for every world $v\in W\cap W'$ and,  for all $i \ge 0$, $V'(u_i):=V(x)$, with $x$ a maximal world in $\peq(S^i(w))$, it can be verified that the relation $\bisim$ displayed in terms of red dashed lines is a bisimulation between $\model$ and $\model'$. The reflexivity and transitivity of $\peq$ and $\peq'$ is not represented for the sake of readability.}
		\label{fig:same:maximal}
	\end{figure}
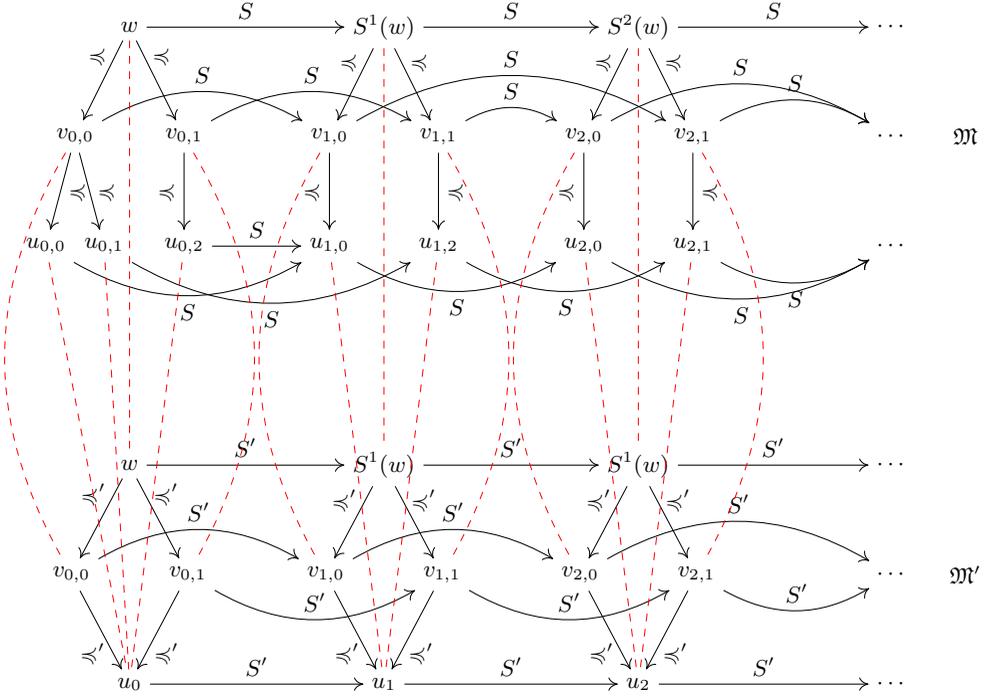

Let us define $V'$ as $V'(v):= V(v)$ if $v \not = u_i$ and $V'(u_i) := V(x)$ with $x$ being any maximal world in $\peq(S^i(w))$\footnote{Note that, since $\model$ is a \itlbd{n}  model, such $x$ always exists. Moreover, since by assumption $\model, S^i(w) \models \lbrace \neg p \vee \neg \neg p \mid p \in \PV \rbrace$, all $\peq$-maximal worlds in $\peq(S^i(w))$ satisfy the same propositions.}, for all $i \ge 0$.
We define the relation $\bisim\subseteq W\times W'$ as follows:
\begin{eqnarray*}
	\bisim & := & \lbrace (v,v) \mid v\in W \hbox{ and } v \in W'\rbrace \\
	&& \cup \lbrace (v,u_i) \mid v \in W, S^i(w)\peq v \hbox{ and } v \hbox{ is maximal w.r.t.} \peq\rbrace.
\end{eqnarray*}
\noindent It can be checked that $\bisim$ is a bisimulation.
Figure~\ref{fig:same:maximal} provides an example of an \itlbd{n} model $\model$ and its bisimilar model $\model'$ owning an unique $\peq'$-maximal world $u_i$ per time instant.
\end{proof}

\begin{lemma}[Contraction lemma for \itlbd{n}]\label{lem:contraction:itl}Let $\model = \tuple{(W,\peq, S),V}$ be an \itlbd{n} model and let $w\in W$.
	If for all $i\ge0$
	\begin{enumerate}
		\item \label{contraction:itl:cond0} there exists an unique maximal world, denoted by $u_i$, in $\peq(S^i(w))$ and
		\item \label{contraction:itl:cond1} $V(v) = T_i \subseteq \PV$ for all $S^i(w) \prec v$
	\end{enumerate}
        then there exists a \THT\ model
	$\model'=\tuple{(\mathbb{N}\times\lbrace 0,1\rbrace,\peq',S'),V'}$ and a bisimulation $\bisim \subseteq W \times (\mathbb{N}\times\lbrace 0,1\rbrace)$
	such that, for all $i\ge 0$, $S^i(w) \bisim S'^i((0,0))$ and $S^i(u_0) \bisim S'^i((0,1))$.
\end{lemma}
\begin{proof}
	Let us define the \THT\ model $\model'=\tuple{(\mathbb{N}\times\lbrace 0,1\rbrace,\peq',S'),V'}$
	where $V'$ is defined as $V'((i,0)) := V(S^i(w))$ and $V'((i,1)) := V(u_i)$, for all $i \ge 0$\footnote{All $u_i$s exist because of Condition~\ref{contraction:itl:cond0}.}.
	Let us define now the relation $\bisim\subseteq W \times (\mathbb{N}\times\lbrace 0,1\rbrace)$ as
		\begin{equation}
			\bisim := \lbrace (S^i(w),(i,0)), (v, (i,1)) \mid i \ge 0 \hbox{ and }  v \in \prec(S^i(w)) \rbrace.
		\end{equation}

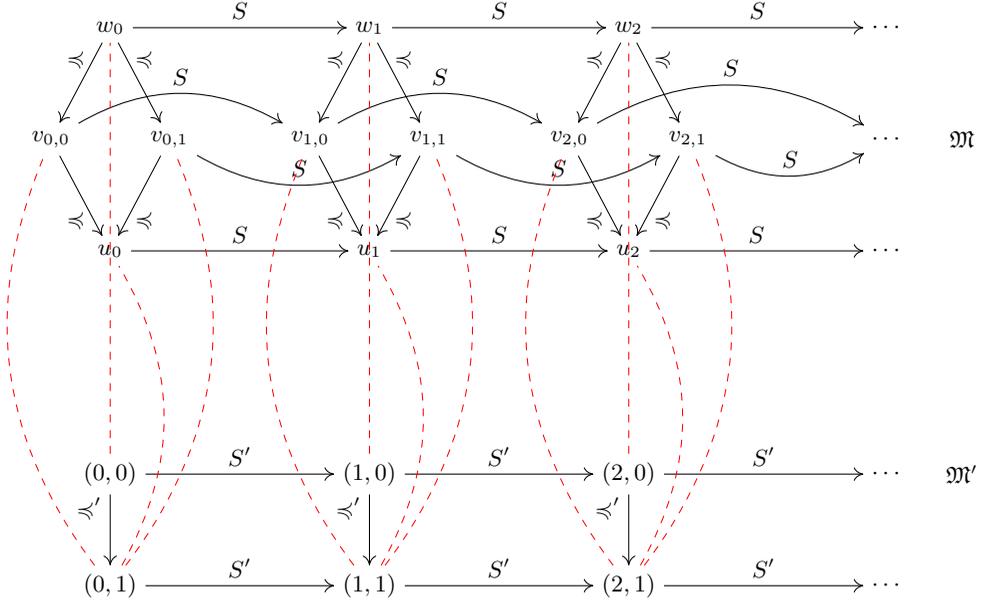
\begin{figure}[h!]\centering
	\resizebox{\textwidth}{!}{
		\begin{tikzpicture}
			[level distance=1.5cm,
			level 1/.style={sibling distance=1.5cm},
			level 2/.style={sibling distance=0.8cm}]

			\node (wp1) {$w_{0}$}
			child[->,sibling distance=.8cm] { node (vp1) {$v_{0,0}$}  edge from parent node[left,pos=.2] {$\peq$}}
			child[->]{ child {node (up1) {$u_0$} edge from parent[draw=none]}edge from parent[draw=none]}
			child[->,sibling distance=.8cm] { node (vp2) {$v_{0,1}$}  edge from parent node[right,pos=.2] {$\peq$}};
			\draw[->] (vp1) --  node[left,pos=.8] {$\peq$}(up1);
			\draw[->] (vp2) --  node[right,pos=.8] {$\peq$}(up1);

			\node[right of=wp1, node distance=3.5cm] (wwp1) {$w_{1}$}
			child[->,sibling distance=.8cm] { node (vvp1) {$v_{1,0}$}
				edge from parent node[left,pos=.2] {$\peq$}
			}
			child[->]{ child {node (uup1) {$u_1$} edge from parent[draw=none]}edge from parent[draw=none]}
			child [->,sibling distance=.8cm] {node (vvp2) {$v_{1,1}$}
				edge from parent node[right,pos=.2] {$\peq$}
			};
			\draw[->] (vvp1) -- node[left,pos=.8] {$\peq$}(uup1);
			\draw[->] (vvp2) -- node[right,pos=.8] {$\peq$}(uup1);

			\node[right of=wwp1, node distance=3.5cm] (wwwp1) {$w_{2}$}
			child[->,sibling distance=.8cm] {node (vvvp1) {$v_{2,0}$}
				edge from parent node[left,pos=.2] {$\peq$}
			}
			child[->]{ child[->] {node (uuup1) {$u_2$} edge from parent[draw=none]}edge from parent[draw=none]}
			child[->,sibling distance=.8cm] {node (vvvp2) {$v_{2,1}$}
				edge from parent node[right,pos=.2] {$\peq$}
			};
			\draw[->] (vvvp1) -- node[left,pos=.8] {$\peq$}(uuup1);
			\draw[->] (vvvp2) -- node[right,pos=.8] {$\peq$}(uuup1);

			\node[right of=wwwp1, node distance=3.5cm] (wwwwp1) {$\cdots$}
			child {node (vvvvp1) {$\cdots$} edge from parent[draw=none]
				child {node (uuuup1) {$\cdots$} edge from parent[draw=none]}
			};

			\path[->]
			(wp1) edge[] node[above]{$S$}(wwp1)
			(wwp1) edge[] node[above]{$S$}(wwwp1)
			(wwwp1) edge[] node[above]{$S$}(wwwwp1)

			(vp1) edge[bend left] node[above]{$S$}(vvp1)
			(vvp1) edge[bend left] node[above]{$S$}(vvvp1)
			(vvvp1) edge[bend left] node[above]{$S$}(vvvvp1)

			(vp2) edge[bend right] node[above]{$S$}(vvp2)
			(vvp2) edge[bend right] node[above]{$S$}(vvvp2)
			(vvvp2) edge[bend right] node[above]{$S$}(vvvvp1)

			(up1) edge[] node[above]{$S$}(uup1)
			(uup1) edge[] node[above]{$S$}(uuup1)
			(uuup1) edge[] node[above]{$S$}(uuuup1)
			;

			\node[below of=up1,node distance=3cm] (wq1) {$(0,0)$}
				child[->]{node[] (uq1) {$(0,1)$} 	edge from parent node[left,pos=.2] {$\peq'$}};

			\node[right of=wq1, node distance=3.5cm] (wwq1) {$(1,0)$}
				child[->] {node (uuq1) {$(1,1)$} 	edge from parent node[left,pos=.2] {$\peq'$}};

			\node[right of=wwq1, node distance=3.5cm] (wwwq1) {$(2,0)$}

			child[->] {node (uuuq1) {$(2,1)$} 	edge from parent node[left,pos=.2] {$\peq'$}};

			\node[right of=wwwq1, node distance=3.5cm] (wwwwq1) {$\cdots$}
			child {node (uuuuq1) {$\cdots$} edge from parent[draw=none]};

			\node[right of=vvvvp1]{$\model$};
			\node[right of=wwwwq1]{$\model'$};

			\path[->]
				(wq1) edge[] node[above]{$S'$}(wwq1)
				(wwq1) edge[] node[above]{$S'$}(wwwq1)
				(wwwq1) edge[] node[above]{$S'$}(wwwwq1)

				(uq1) edge[] node[above]{$S'$}(uuq1)
				(uuq1) edge[] node[above]{$S'$}(uuuq1)
				(uuuq1) edge[] node[above]{$S'$}(uuuuq1);

			\path[dashed,color=red]
			(wq1)   edge[] (wp1)
			(wwq1)  edge[] (wwp1)
			(wwwq1) edge[] (wwwp1)

			(uq1)   edge[bend left] (vp1)
			(uuq1)  edge[bend left] (vvp1)
			(uuuq1) edge[bend left] (vvvp1)

			(uq1)   edge[bend right] (vp2)
			(uuq1)  edge[bend right] (vvp2)
			(uuuq1) edge[bend right] (vvvp2)

			(uq1) edge[bend right] (up1)
			(uuq1) edge[bend right] (uup1)
			(uuuq1) edge[bend right] (uuup1);
		\end{tikzpicture}
	}
	\caption{Example of model contraction. $\model=\tuple{(W,\peq,S),V}$ is an \itlbd{n} model and
		$\model'=\tuple{(\mathbb{N}\times\lbrace 0,1\rbrace,\peq',S'),V'}$ is a \THT\ model.
		Under the assumption that every world $v\in \prec(S^i(w))$ satisfies exactly the same set of propositional variables,
		we can set $V'((i,0)):= V(w_i)$ and $V'((i,1)) := V(u_i)$, for all $i \ge 0$.
		The relation $\bisim$, displayed in red dashed lines, is a bisimulation between $\model$ and $\model'$.
		The reflexivity and transitivity of $\peq'$ and $\peq$ is not represented for the sake of readability.}
	\label{fig:contraction:itl}
\end{figure}

It can be checked that $\bisim$ is an intuitionistic temporal bisimulation between $\model$ and $\model'$.
The condition for the propositional variables is satisfied because of Condition~\ref{contraction:itl:cond1}.
The other conditions can be easily checked.
Figure~\ref{fig:contraction:itl} shows an example of how a bisimulation between an \itlbd{n} model $\model$ and a \THT\ model $\model'$, which can be constructed whenever $\model$ satisfies the preconditions~\ref{contraction:itl:cond0} and~\ref{contraction:itl:cond1} stated in this lemma.
\end{proof}

\subsection{Temporal Equilibrium Logic}\label{sec:tel}

Given two \THT\ models $\model' = \tuple{(W,\peq,S),V'}$ and $\model = \tuple{(W,\peq,S),V}$,
we define $\model' \le \model$ if $V'((i,1)) = V((i,1))$ and $V'((i,0)) \subseteq V((i,0))$ for all $i \ge 0$,
and $\model' = \model$ if $V'((i,x)) = V((i,x))$ for all $i\ge 0$ and $x \in \lbrace 0,1\rbrace$.
Strict inequality $\model' < \model$ is defined as $\model' \le \model$ and $\model' \neq \model$.
Finally, we also say that $\model$ is \emph{total} if $V((i,0))=V((i,1))$, for all $i \ge 0$.

The following result is a corollary of Proposition~\ref{PropIntCond}.
\begin{corollary}[Satisfaction of negation]\label{cor:negation}
For any \THT\ model $\model = \tuple{(W,\peq,S),V}$, for any $i \ge 0$ and for all $\varphi \in \langT$,  $\model,(i,0) \models \neg \varphi$ iff $\model, (i,1) \not \models \varphi$
\end{corollary}
\begin{definition}
  We say that a total \THT\ model $\model=\tuple{(W,\peq,S),V}$ is an equilibrium logic of a temporal formula $\varphi$ if
  \begin{enumerate}
  \item $\model, (0,0) \models \varphi$ and
  \item there is no \THT\ model $\model'$ such that $\model' < \model$ and $\model', (0,0) \models \varphi$.
  \end{enumerate}
\end{definition}
\emph{Temporal Equilibrium Logic} (\TEL\ for short) is the nonmonotonic logic induced by the temporal equilibrium models.
 \section{Two Fixpoint Characterisations of Temporal Equilibrium Logic}\label{sec:main-contribution}
In this section,
we extend the fixpoint characterizations presented in Section~\ref{sec:fixpoint} to the temporal case.
In order to extend Pearce's characterization to the \TEL\ case, we need to reformulate some of his definitions.
In this section, given a \THT\ model $\model$, we \emph{redefine}
\begin{displaymath}
  \Th{\model} := \lbrace \varphi \mid  \model, (0,0) \Inf{\THT} \varphi \rbrace.
\end{displaymath}

\begin{proposition}\label{prop:total}
  Let $\model= \tuple{(W,\peq, S), V}$ be a temporal equilibrium model of $\Gamma$.
  For every \THT\  model $\model'=\tuple{(W,\peq,S),V'}$,
  if $\model', (0,0) \models \Gamma \cup \lbrace \neg \varphi \mid \varphi \not \in \Th{\model}\rbrace$,
  then $V((i,1))=V'((i,1))=V'((i,0))$ for all $i \ge 0$.
\end{proposition}
\begin{proof} Assume towards a contradiction that $\model', (0,0) \models \Gamma \cup \lbrace \neg \varphi \mid \varphi \not \in \Th{\model}\rbrace$ but there exists $i\ge 0$ such that not $V((i,1))=V'((i,1))=V'((i,0))$.
	We first consider the case where $V((i,1)) \not= V'((i,1))$. There are two cases:

	\begin{itemize}
		\item If  $V((i,1)) \not \subseteq V'((i,1))$, there exists some $p \in V((i,1))$ such that $p \not \in V'((i,1))$.
		Since $p \in V((i,1))$, then $\model, (0,1)\models \next^i p$. Since $\model$ is a total model, it follows that  $\model, (0,0)\models \next^i p $ and $\model, (0,0)\not \models \neg \next^i p $.
		Therefore, $\neg \next^i p \not \in \Th{\model}$.
		Since $\model', (0,0) \models \Gamma \cup \lbrace \neg \varphi \mid \varphi \not \in \Th{\model}\rbrace$, then $\model',(0,0)\models \neg \neg \next^i p$. By Proposition~\ref{PropIntCond}, $\model', (0,1) \models \neg \neg \next^i p$. Since the world $(0,1)$ is a classical world, $\model', (0,1) \models \next^i p$ so $p \in V'((i,1))$: a contradiction.

		\item If $V((i,1)) \not\supseteq V'((i,1))$, there exists some $p \in V'((i,1))$ such that $p \not \in V((i,1))$.
		Since $p \not \in V((i,1))$, then $\model, (0,1) \not \models \next^i p$.
		By Proposition~\ref{PropIntCond}, $\model,(0,0) \not \models \next^i p$.
		Therefore, $\next^i p \not \in \Th{\model}$.
		Since $\model', (0,0) \models \Gamma \cup \lbrace \neg \varphi \mid \varphi \not \in \Th{\model} \rbrace$ then $\model',(0,0)\models \neg \next^i p$.
		By the satisfaction relation it follows that $\model', (0,1) \not \models \next^i p$, so $p \not \in V'((i,1))$: a contradiction.
	\end{itemize}

	Therefore, we can assume that $V((i,1)) = V'((i,1))$, for all $i \ge 0$.
	For the case, $V'((i,0)) \not = V'((i,1))$, we can conclude that $V'((i,0)) \subset V'((i,1))$.
	Therefore, $\model' < \model$. Since $\model$ is a temporal equilibrium model of $\Gamma$ then, $\model', (0,0)\not \models \Gamma$, so $\model',(0,0) \not \models  \Gamma \cup \lbrace \neg \varphi \mid \varphi \not \in \Th{\model}\rbrace$: a contradiction.
\end{proof}

In the temporal case,
we can obtain the same result by replacing \HT\ for \THT\ as underlying logic, as stated in the following proposition.
\begin{lemma}\label{lemma:fixpoint}
  For any theory $\Gamma$ and any total \THT\ model $\model$, the following items are equivalent:
  \begin{enumerate}[label=\arabic*)]
  \item\label{lemma:fixpoint:1} $\model$ is a temporal equilibrium model of $\Gamma$
  \item\label{lemma:fixpoint:2} $\Gamma \cup \lbrace \neg \varphi \mid \varphi \not \in \Th{\model}\rbrace \Inf{\THT} \varphi $ iff $\varphi \in \Th{\model}$ for all $\varphi \in \langT$.
  \end{enumerate}
\end{lemma}
\begin{proof}
	To prove that Item~\ref{lemma:fixpoint:1} implies Item~\ref{lemma:fixpoint:2} we assume that Item~\ref{lemma:fixpoint:1} holds but~\ref{lemma:fixpoint:2} does not.
	Then, $\model$ is a temporal equilibrium model of $\Gamma$ but there exists a formula $\varphi \in \langT$ for which one of the following two cases hold:

	\begin{itemize}
		\item $\Gamma \cup \lbrace \neg \varphi \mid \varphi \not \in \Th{\model}\rbrace \Inf{\THT} \varphi$ but $\varphi \not \in \Th{\model}$: in this case, since $\model$ is a temporal equilibrium model of $\Gamma$ then $\model$ is total and, in addition, $\model,(0,0) \models \Gamma$. We can easily check that $\model, (0,0) \models \lbrace \neg \varphi \mid \varphi \not \in \Th{\model}\rbrace$.
		Therefore, $\model,(0,0) \models \varphi$ which contradicts  $\varphi \not \in \Th{\model}$.

		\item $\varphi \in \Th{\model}$ but $\Gamma \cup \lbrace \neg \varphi \mid \varphi \not \in \Th{\model}\rbrace \not \Inf{\THT} \varphi$: in this case, there exists $\model' = \tuple{(W,\peq,S),V'}$ such that $\model', (0,0) \models \Gamma \cup \lbrace \neg \varphi \mid \varphi \not \in \Th{\model}\rbrace$ but $\model', (0,0) \not \models \varphi$.
		From  $\model', (0,0) \models \Gamma \cup \lbrace \neg \varphi \mid \varphi \not \in \Th{\model}\rbrace$ and Proposition~\ref{prop:total} it follows $V' = V$. Therefore, $\model, (0,0) \not \models \varphi$, which means that $\varphi \not \in \Th{\model}$: a contradiction.
  \end{itemize}

	For the converse direction, let us assume towards a contradiction that  $\model$ is not an equilibrium model of $\Gamma$.
	We assume without loss of generality that $\model$ is total but one of the following conditions fails.

	\begin{itemize}
		\item $\model, (0,0) \not \models \Gamma$. Assume that $\Gamma \not = \emptyset$ so there exists $\varphi \in \Gamma$ such that $\model , (0,0) \not \models \varphi$.
		This means that $\varphi \not \in \Th{\model}$. Since item~\ref{lemma:fixpoint:2} holds, $\Gamma \cup \lbrace \neg \varphi \mid \varphi \not \in \Th{\model}\rbrace \not\Inf{\THT} \varphi$.
		It follows that there exists
		$\model'=\tuple{(W,\peq,S),V'}$ such that $\model',(0,0) \models \Gamma \cup \lbrace \neg \varphi \mid \varphi \not \in \Th{\model}\rbrace$ but $\model', (0,0) \not \models \varphi$.
		Since $\model', (0,0) \models \Gamma$ and $\varphi \in \Gamma$ then $\model',(0,0) \models \varphi$.
		Since, $\varphi\not \in \Th{\model}$ then $\model', (0,0)\models \neg \varphi$.
		From the two previous points we conclude that $\model', (0,0)\models \bot$: a contradiction.

		\item $\model, (0,0) \models \Gamma$ but there exists $\model' = \tuple{(W,\peq,S),V'}$ such that $\model' < \model$ and $\model', (0,0) \models \Gamma$. From $\model' < \model$ follows that there exists $i \ge 0$ and $\next^i p \in \langT$ such that $\model', (0,0) \not \models \next^i p$, but $\model, (0,0) \models \next^i p$.

		Since $\next^i p \in \Th{\model}$ then
		$\Gamma \cup \lbrace \neg \varphi \mid \varphi \not \in \Th{\model}\rbrace \Inf{\THT} \next^i p$.
		It can be checked that $\model',(0,0) \models \lbrace \neg \varphi \mid \varphi \not \in \Th{\model}\rbrace$. Therefore, $\model', (0,0) \models \next^i p$, a contradiction.
	\end{itemize}
\end{proof}

\subsection{Temporal safe beliefs}

For extending Definition~\ref{def:safe-belief} to the temporal case,
we need some extra definitions.
Since in the temporal case the truth of an atom may vary along time,
we define the so-called set of \emph{temporal atoms} associated with a signature $\PV$ (in symbols, $\TPV$) as
\begin{displaymath}
  \TPV := \lbrace \next^i p \mid p \in \PV \hbox{ and } i \ge 0 \rbrace.
\end{displaymath}
Clearly, for any $p \in \PV$, $\next^0 p := p$, so $\PV\subseteq \TPV$.
\begin{definition}[$\itlbd{n}$-temporal safe belief]\label{def:temporal:safe-belief}
  Let $\Gamma$ be a temporal theory.
  The set $T \subseteq \TPV$ is said to be a $\itlbd{n}$-\emph{temporal safe belief set} with respect to $\Gamma$ if
  \begin{enumerate}
  \item $\Gamma \cup \lbrace \next^i \neg \neg  p \mid \next^i p \in T \rbrace \cup \lbrace \next^i \neg  p \mid \next^i p \not \in T\rbrace$ is $\itlbd{n}$-consistent  and
  \item\label{safe:belief:cond:2} $\Gamma \cup \lbrace  \next^i \neg \neg p \mid \next^i p \in T \rbrace \cup \lbrace \next^i \neg  p \mid  \next^i p  \not \in T \rbrace \Inf{\itlbd{n}} T$.\qed
  \end{enumerate}
\end{definition}
In the definition above,
$\itlbd{n}$ can be exchanged by any other proper intermediate temporal logic $\logic{X}$.
In the particular case of \THT,
we can prove a correspondence between \THT-temporal safe beliefs and temporal equilibrium models.
\begin{definition}\label{def:correspondence}
  Given a total \THT\ model $\model=\tuple{(W,\peq,S),V}$ we define
  \begin{displaymath}
    T := \lbrace \next^i p \mid p \in V((i,0)) \hbox{ and } i \ge 0 \rbrace.
  \end{displaymath}
  Clearly, $T \subseteq \TPV$.
  Conversely, given $T$ we retrieve $\model=\tuple{(W,\peq,S),V}$ by setting
  \begin{displaymath}
    V((i,0)),V((i,1)):=\lbrace p \mid \next^i p \in T\rbrace,
    \text{ for all } i \ge 0.
  \end{displaymath}
\end{definition}
\begin{proposition}
  For any temporal theory $\Gamma$,
  any total \THT\ model $\model= \tuple{(W,\peq,S),V}$ and
  set $T \subseteq \TPV$
  related as described in Definition~\ref{def:correspondence},
  the following items are equivalent:
  \begin{enumerate}[label=\arabic*)]
  \item\label{prop:a} $\model$ is a temporal equilibrium model of $\Gamma$
  \item\label{prop:b} $T$ is a \THT-temporal safe belief of $\Gamma$
  \end{enumerate}
\end{proposition}
\begin{proof}

	To prove that~\ref{prop:a} implies~\ref{prop:b}, let us assume that $T$ is not a \THT-temporal safe belief of $\Gamma$.
	Let us assume that
	\begin{equation*}
	\Gamma \cup \lbrace \neg \neg \next^i p \mid \next^i p \in T \rbrace \cup \lbrace \neg \next^i p \mid \next^i p \not \in T  \rbrace
	\end{equation*}
        is consistent but
	\begin{equation*}
	\Gamma \cup \lbrace \neg \neg \next^i p \mid \next^i p \in T\rbrace \cup \lbrace \neg \next^i p \mid \next^i p \not \in T \rbrace\not \Inf{\THT} T.
	\end{equation*}
	This means that there exists a \THT~model $\model' = \tuple{(W, \peq ,S), V'}$ such that $\model', (0,0) \models \Gamma$, $\model', (0,0) \models \lbrace \neg \neg \next^i p \mid \next^i p \in T\rbrace$, $\model', (0,0) \models \lbrace \neg \next^i p \mid \next^i p \not \in T\rbrace$ but $\model', (0,0) \not \models T$.
	From $\model', (0,0) \models \lbrace \neg \neg \next^i p \mid \next^i p \in T\rbrace$ and  $\model', (0,0) \models \lbrace \neg \next^i p \mid \next^i p \not \in T\rbrace$ we can conclude that
	$V((i,1)) = V'((i,1))$, for all $i \ge 0$.

	From $\model', (0,0) \not \models T$ it follows that $\model', (0,0)\not \models \next^i p$ for some $\next^i p \in T$, with $i \ge 0$.
	This means that $\model', (i,0) \not \models p$.
	From $\model', (0,0) \models \lbrace \neg \neg \next^i p \mid \next^i p \in T\rbrace$ we conclude that $\model', (0,0) \models \neg \neg \next^i p$, so $\model', (i,1) \models p$.
	Therefore, $\model' < \model$.
	Since $\model', (0,0) \models \Gamma$, $\model$ is not an equilibrium model of $\Gamma$: a contradiction.

	Conversely, let us assume towards a contradiction that $T$ is a \THT-temporal safe belief of $\Gamma$ but $\model = \tuple{(W,\peq,S), V}$ is not a temporal equilibrium model of $\Gamma$.
	Assume, without loss a contradiction that $\model$ is total.
	Since $T$ is a \THT-temporal safe belief of $\Gamma$ then $\Gamma \cup \lbrace \neg \neg \next^i p \mid \next^i p \in T \rbrace \cup \lbrace \neg \next^i p \mid \next^i p \not \in T  \rbrace$ is consistent.
	Let $\model'=\tuple{(W,\peq,S),V'}$ be such that $\model', (0,0) \models \Gamma \cup \lbrace \neg \neg \next^i p \mid \next^i p \in T \rbrace \cup \lbrace \neg \next^i p \mid \next^i p \not \in T  \rbrace$.
	Since $\model', (0,0) \models \lbrace \neg \neg \next^i p \mid \next^i p \in T \rbrace \cup \lbrace \neg \next^i p \mid \next^i p \not \in T\rbrace$ then $V'(i,1)=V(i,1)$, for all $i \ge 0$.
	Since $\model', (0,0) \models \Gamma$ then $\model',(0,1)\models \Gamma$.
	Since $\model$ is total and $V'(i,1)=V(i,1)$, for all $i \ge 0$ then $\model, (0,0)\models \Gamma$.
	Since $\model$ is not an equilibrium model of $\Gamma$, there exists $\model''=\tuple{(W,\peq,S),V''}$ such that $\model''< \model$ and $\model'',(0,0)\models \Gamma$.
	However, this contradicts  Condition~\ref{safe:belief:cond:2} of Definition~\ref{def:temporal:safe-belief}.
\end{proof}

To conclude this section,
we show that temporal safe belief sets are preserved when changing the underlying logic.
In this case, we extend the results shown in Section~\ref{sec:fixpoint}.
\begin{lemma}\label{prop:itl:intermediate1}
  Let us consider $T\subseteq \TPV$ and
  let $\logic{X}$ and $\logic{Y}$ be two proper intermediate temporal logics satisfying $\logic{X}\subseteq \logic{Y}$.
  For any temporal theory $\Gamma$,
  if $T$ is a $\logic{X}$-temporal safe belief of $\Gamma$, then $T$ is a $\logic{Y}$-temporal safe belief of $\Gamma$.
\end{lemma}
\begin{proof}

		If $T$ is a $\logic{X}$-temporal safe belief of $\Gamma$, it follows that
		\begin{enumerate}
			\item $\Gamma \cup \lbrace  \next^i \neg \neg p \mid \next^i p \in T \rbrace \cup \lbrace \next^i \neg  p \mid \next^i p \not \in T\rbrace$ is $\logic{X}$-consistent and
			\item $\Gamma \cup \lbrace  \next^i \neg \neg p \mid \next^i p \in T \rbrace \cup \lbrace \next^i \neg  p \mid \next^i p \not \in T\rbrace\Inf{\logic{X}} T$
		\end{enumerate}
		From the first item and Proposition~\ref{prop:consistency:itl} it follows that
		\begin{displaymath}
		\Gamma \cup \lbrace  \next^i \neg \neg p \mid \next^i p \in T \rbrace \cup \lbrace \next^i \neg p \mid \next^i p \not \in T\rbrace
		\end{displaymath}
                is $\logic{Y}$-consistent. From the second item and Proposition~\ref{prop:entailment:itl} we conclude that
		\begin{displaymath}
			\Gamma \cup \lbrace \next^i \neg \neg p \mid \next^i p \in T \rbrace \cup \lbrace \next^i \neg p \mid \next^i p \not \in T\rbrace\Inf{\logic{Y}} T.
		\end{displaymath}
		As a consequence, $T$ is a $\logic{Y}$-safe belief of $\Gamma$.
\end{proof}
We prove the converse of Lemma~\ref{prop:itl:intermediate1} below.
\begin{lemma}\label{prop:itl:intermediate2}
  Let us consider $T\subseteq \TPV$ and
  let $\logic{X}$ and $\logic{Y}$ be two proper intermediate temporal logics satisfying $\logic{X}\subseteq\logic{Y}$.
  For any temporal theory $\Gamma$,
  if $T$ is a $\logic{Y}$-temporal safe belief set of $\Gamma$,
  then $T$ is a $\logic{X}$-temporal safe belief set of $\Gamma$.
\end{lemma}
\begin{proof}

	Let us assume that $T$ is a $\logic{Y}$-temporal safe belief of $\Gamma$. It holds that
	\begin{enumerate}[label=(\alph*)]
		\item $\Gamma \cup \lbrace  \next^i \neg \neg p \mid \next^i p \in T \rbrace \cup \lbrace \next^i \neg p \mid \next^i p \not \in T\rbrace$ is $\logic{Y}$-consistent and
		\item $\Gamma \cup \lbrace  \next^i \neg \neg p \mid \next^i p \in T \rbrace \cup \lbrace  \next^i \neg p \mid \next^i p \not \in T\rbrace\Inf{\logic{Y}} T$.
	\end{enumerate}

        From item (a) and Proposition~\ref{prop:consistency:itl} it follows that
	\begin{displaymath}
	\Gamma \cup \lbrace \next^i \neg \neg p \mid p \in T \rbrace \cup \lbrace \next^i \neg p \mid p \not \in T\rbrace
	\end{displaymath}
	is both $\logic{X}$-consistent and $\itlbd{n}$-consistent.
	From $\logic{Y}\subseteq \logic{\THT}$, the second item and Proposition~\ref{prop:entailment:itl}
	\begin{equation}
		\Gamma \cup \lbrace  \next^i \neg \neg p \mid \next^i p \in T \rbrace \cup \lbrace  \next^i \neg p \mid \next^i p \not \in T\rbrace\Inf{\THT} T. \label{tht:cons}
	\end{equation}
	Let $\model=\tuple{(W,\peq,S),V}$ be an $\itlbd{n}$ model and $w \in W$ satisfying
	\begin{enumerate}[label=(\alph*),start=3]
		\item\label{prop:itl:intermediate2:c1} $\model, w \Inf{\itlbd{n}} \Gamma$,
		\item\label{prop:itl:intermediate2:c2} $\model, w \Inf{\itlbd{n}} \lbrace  \next^i \neg \neg p \mid \next^i p \in T \rbrace $ and
		\item\label{prop:itl:intermediate2:c3} $\model, w \Inf{\itlbd{n}} \lbrace \next^i \neg p \mid \next^i p \not \in T\rbrace $.
\end{enumerate}
	From items~\ref{prop:itl:intermediate2:c2} and~\ref{prop:itl:intermediate2:c3} it follows that $\model, w \models \lbrace \alwaysF(\neg \neg p \vee \neg p) \mid p \in \PV \rbrace$.

	In view of Lemma~\ref{lem:bisim:itl:topwidth1}, there exists an \itlbd{n} model $\model' = \tuple{(W',\peq',S'),V'}$, $w' \in W'$ and a bisimulation $\bisim \subseteq W \times W'$ such that, $w\bisim w'$ and for each $i \ge 0$, the intuitionistic subframe generated by $S'^i(w')$ contains an unique maximal point that we name $u'_i$.
By Lemma~\ref{lem:bisim:itl}, $\model', w' \Inf{\itlbd{n}} \Gamma$, $\model', w' \Inf{\itlbd{n}} \lbrace  \next^i \neg \neg p \mid \next^i p \in T \rbrace $ and $\model', w' \Inf{\itlbd{n}} \lbrace \next^i \neg p \mid \next^i p \not \in T\rbrace $.
We prove the following claim:
	\begin{equation}
		\hbox{For all $v' \in W'$, if $w'\peq v'$ then }	\model', v' \models  T. \label{prop:itl:intermediate2:claim}
	\end{equation}
	The proof is done by induction in $\depth{((W',\peq'),v')}$.

	\begin{enumerate}
		\item If $\depth{((W',\peq'),v')}=1$ then $v'$ is maximal, so $v'=u'_0$.
		By monotonicity,
\begin{displaymath}
			\model', u'_0 \models  \lbrace \next^i \neg \neg p \mid \next^i p \in T \hbox{ and } i\ge 0 \rbrace \cup \lbrace \next^i \neg p \mid \next^i p  \not \in T \hbox{ and } i\ge0\rbrace.
		\end{displaymath}
		Because of Proposition~\ref{prop:maximal:next}, $S^i(u'_0)$ is $\peq$-maximal, for all $i \ge 0$.
		Therefore,
		\begin{displaymath}
			\model', S'^i(u'_0) \models \lbrace  p \mid \next^i p \in T \rbrace  \cup \lbrace  \neg p \mid \next^i p \not \in T\rbrace,
		\end{displaymath}
		for all $i\ge 0$. Hence, $V'(S'^i(u'_0)) = \lbrace p \mid \next^i p \in T\rbrace$, for all $i\ge 0$ so $\model', u'_0 \models T$.

		\item For the inductive step, let us assume that $\depth{((W,\peq'),v')} = n+1$ and the claim holds for every $x \in W'$ satisfying $v'\peq x$ and $\depth{((W,\peq'),x)} \le n$ (so $v' \prec x$).
		By Proposition~\ref{prop:depth:next}, $\depth{((W,\peq),S'^i(v'))} \le n+1$, for all $i\ge 0$.

		By induction hypothesis, for all $x \in \prec'(v')$, $\model', x \models T$.  Moreover, by monotonicity, for all $x \in \peq'(v')$ (including $v'$ itself), $\model', x \models \lbrace \next^i \neg p \mid \next^i p \not \in T\rbrace$.

		By the semantics, for all $i\ge 0$ and for all $y \peq'(S'^i(v'))$, $\model', y \models \lbrace \neg p \mid \next^i p \not \in T\rbrace$.

		By the semantics, for all $i\ge 0$ and for all $y \prec'(S'^i(v'))$, $\model', y \models \lbrace p \mid \next^i p  \in T\rbrace$.
		From the two previous results, it follows that for all $i\ge 0$ and for all $y \in \prec'(S'^i(v'))$, $V'(y) = \lbrace p \mid \next^i p \in T \rbrace$.
		Therefore, every point $y \in \prec'(S'^i(v'))$ satisfies exactly the set $\lbrace p \mid \next^i p \in T\rbrace$.
		By Lemma~\ref{lem:contraction:itl} there exists a \THT\ model $\model''=\tuple{(\mathbb{N}\times \lbrace 0,1\rbrace,\peq'',S''),V''}$ and a bisimulation $\bisim'\subseteq W'\times (\mathbb{N}\times \lbrace 0,1\rbrace)$ such that $v' \bisim' (0,0)$.
		Because of items~\ref{prop:itl:intermediate2:c1}-\ref{prop:itl:intermediate2:c3}, the monotonicity property and Lemma~\ref{lem:bisim:itl},
		\begin{displaymath}
			\model'', (0,0) \Inf{\THT} \Gamma \cup \lbrace  \next^i \neg \neg p \mid \next^i p \in T \rbrace\cup \lbrace \next^i \neg p \mid \next^i p \not \in T\rbrace.
		\end{displaymath}
		Since $\model''$ is a \THT\ model, $\model'',(0,0) \Inf{\THT} T$ because of~\eqref{tht:cons}.
		From $v' \bisim' (0,0)$ and Lemma~\ref{lem:bisim:itl} it follows $\model', v' \Inf{\itlbd{n}} T$.
\end{enumerate}

	As a consequence, $\model', w' \Inf{\itlbd{n}} T$. From $w  \bisim w'$ and Lemma~\ref{lem:bisim:itl}, $\model,w \Inf{\itlbd{n}} T$.
	Since $\model$ was chosen arbitrary it follows that
	\begin{equation*}
		\Gamma \cup \lbrace  \next^i \neg \neg p \mid \next^i p \in T \rbrace \cup \lbrace  \next^i \neg p \mid \next^i p \not \in T\rbrace\Inf{\itlbd{n}} T,
	\end{equation*}
	so $T$ is a \itlbd{n}-safe belief of $\Gamma$. By Lemma~\ref{prop:itl:intermediate1} and the fact that $\itlbd{n}\subseteq \logic{X}$,
	\begin{equation*}
		\Gamma \cup \lbrace  \next^i \neg \neg p \mid \next^i p \in T \rbrace \cup \lbrace  \next^i \neg p \mid \next^i p \not \in T\rbrace\Inf{\logic{X}} T.
	\end{equation*}
\end{proof}

Lemmas~\ref{prop:itl:intermediate1} and~\ref{prop:itl:intermediate2} state that \THT\ can be replaced by any proper intermediate temporal logic extending \itlbd{n} without affecting the resulting safe beliefs.
\begin{theorem}\label{thm:main}
For any intermediate temporal logic $\logic{X}$ satisfying $\itlbd{n} \subseteq \logic{X} \subseteq \THT$ and for any theory $\Gamma$, the set of $\logic{X}$-temporal safe beliefs of $\Gamma$ coincide.
\end{theorem}
 \section{Conclusions}\label{sec:conclusions}
In this paper,
we revisited two well-known fixpoint characterizations of propositional equilibrium logic and answer sets.
The first characterization, originally defined by \cite{pearce06a,pearce99b}
is based on the concept of theory completions,
which has also been used in autoepistemic and default logic~\citep{martru93a,besnard89}.
We extended this characterization to the case of \TEL{}.

The second characterization, introduced by \cite{osnaar05a} and known as safe belief sets,
relates the equilibrium logic of arbitrary theories to syntactic entailment in intuitionistic logic.\footnote{Similar results were also obtained by David Pearce for the case of disjunctive logic programs~\citep{pearce99a}.}
The authors proved that intuitionistic logic can be replaced by any proper intermediate logic without changing the set of safe beliefs.
Their results rely on syntactic transformations that cannot be easily reproduced in the temporal case.
Therefore, as a first contribution, we reformulated Osorio et al.'s approach in terms of semantic consequence in intuitionistic logic.

We have identified a family of intuitionistic temporal logics, $\itlbd{n}$,
for which we have defined a temporal extension of safe beliefs.
We first show that any proper intermediate temporal logic extending $\itlbd{n}$ can be used instead,
without affecting the resulting set of temporal safe beliefs.
Moreover, we show that in the case of \THT{},
temporal safe beliefs correspond to temporal equilibrium models.

We believe our results have fostered connections between temporal answer set programming and constructive modal logic,
while also enhancing the visibility of \THT\ within the field of constructive temporal logics.
In future work,
we plan to investigate intermediate logics not covered here, such as \iltl, \itlb,
and real-valued Gödel temporal logics~\citep{ADFM25}.
Since consistency is not always preserved across these logics,
it remains unclear whether the set of safe beliefs is preserved
when using one of them as the monotonic basis for temporal equilibrium logic.

\vspace{30pt}

\noindent \textbf{Competing interests declaration.} The authors declare none.


\appendix
\section{Proofs}

\begin{proofof}{Proposition~\ref{prop:consequence:intermediate}}
	Assume towards a contradiction that $\Gamma \not \Inf{\logic{X}'} \varphi$.
	This means that there exists a model $\model=\tuple{(W,\peq),V}$ and $w\in W$ such that $\model, w \Inf{\logic{X}'}
	\Gamma$ and $\model, w \not \Inf{\logic{X}'} \varphi$.
	Since $\logic{X} \subseteq \logic{X}'$, $\model$ can be reconsidered within the logic $\logic{X}$,
	meaning that $\model, w \Inf{\logic{X}} \Gamma$ and $\model, w \not \Inf{\logic{X}} \varphi$: a contradiction.
\end{proofof}

\begin{proofof}{Proposition~\ref{prop:weakexmiddle}}
	Let $\model = \tuple{(W,\peq), V}$ be an intuitionistic model and let $w \in W$	be such that $\model, w \models \lbrace \neg p \vee \neg \neg p\mid p \in \PV \rbrace$. 
	It follows that all maximal $\peq$-worlds in the subframe generated by $x$ satisfy the same propositional variables. 

	Let us assume towards a contradiction that $\model, w \models \lbrace \neg p \vee \neg \neg p\mid p \in \PV \rbrace$ but there exist two maximal worlds $u$ and $u'$ in $W$ such that 
	$w \peq u$, $w\peq u'$ but $V(u)\not = V(u')$. 
	Let us assume, without loss of generality, that $V(u)\not \subseteq V(u')$. 
	Therefore, there exists $p \in \PV$ such that $p \in V(u)$ but $p \not \in V(u')$.
	Since $u$ is a maximal world, it is classical, so $\model, u \models p \vee \neg p$, so $\model, u \not \models \neg p$.
	By monotonicity, $\model, w \not \models \neg p$.
	Since $\model, w \models \neg p \vee \neg \neg p$, $\model, w \models \neg \neg p$.
	By monotonicity, $\model, u' \models \neg \neg p$.
	Since $u'$ is also classical, it follows that $\model, u' \models p$: a contradiction.
\end{proofof} 

\begin{proofof}{Lemma~\ref{lem:il:safe-beliefs:auxiliary1}}
	 If $T$ is a $\logic{X}$-safe belief of $\Gamma$, it follows that
	\begin{enumerate}
		\item $\Gamma \cup \lbrace \neg \neg p \mid p \in T \rbrace \cup \lbrace \neg p \mid p \not \in T\rbrace$ is $\logic{X}$-consistent and
		\item $\Gamma \cup \lbrace \neg \neg p \mid p \in T \rbrace \cup \lbrace \neg p \mid p \not \in T\rbrace\Inf{\logic{X}} T$
	\end{enumerate}
	From the first item and Proposition~\ref{lem:consistency},
	it follows that
	$\Gamma \cup \lbrace \neg \neg p \mid p \in T \rbrace \cup \lbrace \neg p \mid p \not \in T\rbrace$
	is $\logic{Y}$-consistent.
	From the second item and Proposition~\ref{prop:consequence:intermediate}, we conclude that
	\begin{displaymath}
		\Gamma \cup \lbrace \neg \neg p \mid p \in T \rbrace \cup \lbrace \neg p \mid p \not \in T\rbrace\Inf{\logic{Y}} T.
	\end{displaymath}
\end{proofof}	 

\begin{proofof}{Proposition~\ref{prop:maximal:next}}
Assume towards a contradiction that $S(w)$ is not maximal w.r.t. $\peq$. 
Therefore there exists $v \in W$ such that $S(w)\prec v$.
By the backward confluence property there exists $u\in W$ such that $w \peq u$ and $S(u)=v$.
Since $w$ is maximal w.r.t. $\peq$ then $w=u$.
this would imply that $S$ is not a function: a contradiction.
\end{proofof}

\begin{proofof}{Lemma~\ref{lem:itlbn2ltl-consistency}}
From right to left, in view of Observation~\ref{observation:ltl}, \LTL{} models are an specific case of $\itlbd{n}$ models so if $\Gamma$ is \LTL{}-consistent 
then it is $\itlbd{n}$-consistent.	
From left to right, if $\Gamma$ is $\itlbd{n}$ consistent then there exists $\itlbd{n}$ model $\model = \tuple{(W,\peq,S),V}$ and $w \in W$ such that $\M, w \models \Gamma$.
Since $\M$ is of finite depth, then there exists a maximal Kripke world $v\in W$ such that $w \peq v$ and $\M, v \models \Gamma$.
Since $w \peq v$ then $S(w)\peq S(v)$ because of the forward confluence property. Moreover, in view of Proposition~\ref{prop:maximal:next}, $S(v)$ is also a maximal point.
Let us define $\model' = \tuple{(W',\peq',S'), V'}$ where $W' = \lbrace S^n(v) \mid n \ge 0 \rbrace$, $\peq' = \lbrace (x,x) \mid x \in W'\rbrace$ and $S'(x) = S(x)$, for all $x \in W'$.
Clearly, $\model'$ is an \LTL{} satisfying $\Gamma$ at $v \in W'$.
\end{proofof}

\begin{proofof}{Lemma~\ref{cor:intermediate2ltl:consistency}}
The right to left direction follows Observation~\ref{observation:ltl} as in the proof of Lemma~\ref{lem:itlbn2ltl-consistency}.
For the left to right direction, note that, since $\logic{X}$ is an intermediate logic,, any model in the logic $\logic{X}$ is also a $\itlbd{n}$ model. 
By Lemma~\ref{lem:itlbn2ltl-consistency}, if $\Gamma$  $\logic{X}$-consistent then $\Gamma$ is $\itlbd{n}$-consistent and, because of Lemma~\ref{lem:itlbn2ltl-consistency},
$\Gamma$ is $\LTL$-consistent. 	
\end{proofof}

\begin{proofof}{Proposition~\ref{prop:depth:next}}
	Assume towards a contradiction that $\depth{((W,\peq),S(w))}> n$.
	Assume without loss of generality that $\depth{((W,\peq),S(w))}= n+1$.
	By definition, there exists $w'_0, w'_1, \cdots, w'_{n+1}$ such that 
	$S(w) = w'_0$ and $w'_i \prec w'_{i+1}$\footnote{We consider $\prec$ because all those $w'_i$ must be different.}, for all $\rangeco{i}{0}{n}$.
	From $S(w) = w'_0 \prec w'_1$ and the backward confluence property it follows that there exists $w_1\in W$ such that $w \peq w_1$ and $S(w_1) = w'_1$.
	From $S(w_1) = w'_1\prec w'_2$ and the backward confluence property, it follows that there exists $w_2\in W$ such that $w_1 \prec w_2$ and $S(w_2) = w'_2$.
	By continuously applying the same reasoning we would conclude that, since $S(w_n) = w'_n\prec w'_{n+1}$, there exists $w_n+1\in W$ such that 
	$w_n\prec w_{n+1}$ and $S(w_{n+1}) = w'_{n+1}$. 
	Therefore, there exists a sequence $w_0,w_1,\cdots,w_{n+1}$ with $w=w_0$ and for all $i\ge 0$, $S(w_i) = w'_i$.
	Since $\depth{((W,\peq),w)}\le n$, $w_j=w_k$ for some $\rangecc{i,k}{0}{n+1}$. 
	Since $S$ is a function, it would mean that $S(w_j) = w'_j$ and $S(w_j) = w'_k$, so $w'_j = w'_k$: a contradiction. 
\end{proofof}

\begin{proofof}{Proposition~\ref{prop:consistency:itl}}
We use Corollary~\ref{cor:intermediate2ltl:consistency} to conclude that $\Gamma$ is $\logic{X}$-consistent iff  $\Gamma$ is $\LTL$-consistent iff $\Gamma$ is $\logic{Y}$-consistent.
\end{proofof}	

\begin{proofof}{Proposition~\ref{prop:entailment:itl}} Assume towards a contradiction that $\Gamma \not \Inf{\logic{y}} \Delta$.
	Therefore, there exists $\model=\tuple{(W,\peq,S),V}$ and $w\in W$ such that $\model, w \Inf{\logic{Y}}\Gamma$ and $\model, w \not \Inf{\logic{Y}} \Delta$.
	Since $\logic{X} \subseteq \logic{Y}$, $\logic{X}$ is weaker than $\logic{Y}$. Therefore, $\model$ is also a model within the logic $\logic{X}$.
	Therefore, $\model, w \Inf{\logic{X}}\Gamma$ and $\model, w \not \Inf{\logic{X}} \Delta$, meaning that $\Gamma \not \Inf{\logic{X}} \Delta$: a contradiction.
\end{proofof}

\begin{proofof}{Theorem~\ref{thm:main}}
	Directly from lemmas~\ref{prop:itl:intermediate1} and~\ref{prop:itl:intermediate2}.
\end{proofof}  
\end{document}